\documentclass[11pt]{article}

\usepackage{amsmath,amsthm,amsfonts,amssymb,anysize,graphicx,verbatim}
\usepackage{hyperref}
\usepackage{dsfont}
\usepackage{enumerate}
\usepackage{color}
\usepackage{fullpage}
\usepackage{xspace}
\usepackage{ifthen}

\newcommand{\ra}{\rightarrow}

\newcommand{\hide}[1]{}
\hide{
     \newcommand{\qed}{\nobreak \ifvmode \relax \else
     \ifdim\lastskip<1.5em \hskip-\lastskip
     \hskip1.5em plus0em minus0.5em \fi \nobreak
     \vrule height0.75em width0.5em depth0.25em\fi}
}

 \newtheorem{theorem}{Theorem}[section]
 \newtheorem{lemma}[theorem]{Lemma}

 \newtheorem{claim}[theorem]{Claim}
  
    \newtheorem{definition}{Definition}[section]

 \newenvironment{pfof}[1]{\begin{proof}[\emph{\emph{Proof of #1: }}]}{\end{proof}}

\newcommand{\pot}{\varphi}
\newcommand{\bpot}{\bar{\varphi} }
\newcommand{\tpot}{\tilde{\varphi} }
\newcommand{\poti}{\varphi_i}

\newcommand{\potj}{\varphi_j}

\newcommand{\bpoti}{\bar{\varphi}_i }
\newcommand{\bpotj}{\bar{\varphi}_j }

\newcommand{\Shtfl}{\text{Shtf}}

\newcommand{\CRi}{\text{\rm CR}_i}

\newcommand{\SR}{\text{\rm SR}}
\newcommand{\SRi}{\text{\rm SR}_i}
\newcommand{\SRj}{\text{\rm SR}_j}

\newcommand{\MR}{\text{\rm MR}}
\newcommand{\EMR}{\overline{\text{\rm MR}}}

\newcommand{\MRi}{\text{\rm MR}_i}
\newcommand{\MRj}{\text{\rm MR}_j}

\newcommand{\exi}{x_i}

\newcommand{\bxi}{\bar{x}_i}
\newcommand{\xii}{\xi_i}
\newcommand{\xij}{\xi_j}
\newcommand{\xie}{x_i^e}
\newcommand{\xis}{x_i^s}
\newcommand{\xil}{x_i^l}

\newcommand{\xijs}{x_{ij}^s}

\newcommand{\xijl}{x_{ij}^l}

\newcommand{\txi}{{\tilde x}_i}
\newcommand{\txj}{{\tilde x}_j}
\newcommand{\txijs}{{\tilde x}^s_{ij}}
\newcommand{\txjis}{{\tilde x}^s_{ji}}
\newcommand{\tli}{t_{l_i}}
\newcommand{\tlip}{t_{l_{i+1}}}
\newcommand{\tlone}{t_{l_1}}
\newcommand{\Ri}{R_i}

\newcommand{\efi}{f_i}
\newcommand{\Fi}{F_i}

\newcommand{\li}{l_i}

\newcommand{\ri}{r_i}
\newcommand{\bri}{{\overline{r}_i}}
\newcommand{\qr}{q(r)}
\newcommand{\qri}{q_{r_i}}
\newcommand{\qrj}{q_{r_j}}
\newcommand{\bqri}{\overline{q}_{r_i}}
\newcommand{\qi}{q_i}
\newcommand{\qj}{q_j}
\newcommand{\nqi}{q_{-i}}
\newcommand{\nqij}{q_{-ij}}

\newcommand{\teeh}{t_h}

\newcommand{\ui}{u_i}

\newcommand{\bq}{\bar q}
\newcommand{\bbq}{\bar \q}

\newcommand{\hbqj}{\hat {\bar q}_j}
\newcommand{\bqp}{{\bar q}\,'}
\newcommand{\bqh}{{\bar q}_h}
\newcommand{\bqi}{{\bar q}_i}
\newcommand{\bqj}{{\bar q}_j}
\newcommand{\qbri}{q_{\bar r_i}}

\newcommand{\bqbri}{{\bar q}_{\bar r_i}}

\newcommand{\bv}{\bar v}
\newcommand{\bvi}{{\bar v}_i}

\newcommand{\dqi}{~dq_i}
\newcommand{\dqj}{~dq_j}
\newcommand{\qalph}{q_{\alph}}
\newcommand{\vi}{v_i}
\newcommand{\vh}{v_h}

\newcommand{\xib}{x_i^b}

\newcommand{\bchii}{{\overline{\chi}}_i}
\newcommand{\chii}{{\chi}_i}
\newcommand{\bgxi}{{\bar \xi}}
\newcommand{\bgxij}{{\bar \xi}_j}
\newcommand{\hgxi}{{\hat \xi}}

\def\bL{\overline{L}}
\def\bR{\overline{R}}
\def\bRi{\overline{R}_i}
\def\bU{\overline{U}}

\def\bCR{\overline{\text{CR}}}
\def\bCRi{\overline{\text{CR}}_i}

\newcommand{\dbydv}[1]{\tfrac {d{#1}} {dv}}

\newcommand{\floor}[1]{\lfloor{#1}\rfloor}

\def\eps{\epsilon}
\def\alph{\alpha}
\def\gam{\gamma}
\def\lmbd{\lambda}
\def\Del{\Delta}
\def\del{\delta}
\def\Pr{\text{\bf Pr}}

\def\Shtf{\text{\rm Shtf}}

\def\Evix{{\mathcal E}_i^x}

\def\eps{\epsilon}

\def\F2{\mathcal{F}^{(2)}}

\def\D{\mathcal{D}}
\def\F{\mathcal{F}}
\def\A{\mathcal{A}}
\def\E{\mathcal{E}}
\def\q{\mathbf{q}}

 \newcommand{\prob}[2][]{\text{\bf Pr}\ifthenelse{\not\equal{}{#1}}{_{#1}}{}\!\left[#2\right]}
 \newcommand{\expect}[2][]{\text{\bf E}\ifthenelse{\not\equal{}{#1}}{_{#1}}{}\!\left[#2\right]}

\newcommand{\val}{v}
\newcommand{\vals}{{\mathbf \val}}

\newcommand{\vali}[1][i]{{\val_{#1}}}

\newcommand{\samples}{\vals^{(1)},\ldots,\vals^{(m)}}
\newcommand{\inp}{\vals^{(m+1)}}

\def\v2{\val^{(2)}}

\newcommand{\dist}{F}
\newcommand{\dists}{{\mathbf \dist}}
\newcommand{\disti}[1][i]{{\dist_{#1}}}

\newcommand{\dens}{f}

\newcommand{\densi}[1][i]{{\dens_{#1}}}

\newcommand{\vv}{\varphi}
\newcommand{\vvi}[1][i]{{\vv_{#1}}}

\newcommand{\poly}{\mbox{poly}}

\title{The Sample Complexity of Revenue Maximization\thanks{A
    preliminary version of this paper 
appeared in the 2014 ACM Symposium on the Theory of Computing.}}

\author{Richard Cole\thanks{This work  was supported in part by NSF
    Awards CCF-1217989 and CCF-1527568.} \\ 
 Courant Institute \\ New York University
\and
Tim Roughgarden\thanks{This research was done in part while visiting
  New York University, and was supported in part by NSF Awards CCF-1016885 and
    CCF-1215965 and an ONR PECASE Award.}\\ Stanford University
}

\begin{document}

\newcommand{\one}{\mathds{1}}

\maketitle

\thispagestyle{empty}

\begin{abstract}
In the design and analysis of revenue-maximizing auctions, auction
performance is typically measured with respect to a prior distribution
over inputs.  The most obvious source for such a distribution is past
data. The goal of this paper is to understand how much data is
necessary and sufficient to guarantee near-optimal expected revenue. 

Our basic model is a single-item auction in which bidders' valuations
are drawn independently from unknown and non-identical
distributions. The seller is given $m$ samples from each of these
distributions ``for free'' and chooses an auction to run on a fresh
sample.  How large does $m$ need to be, as a function of the number
$k$ of bidders and $\eps > 0$, so that a $(1-\eps)$-approximation of
the optimal revenue is achievable? 

We prove that, under standard tail conditions on the underlying
distributions, $m=\poly(k,\tfrac{1}{\eps})$ samples are necessary and
sufficient.  Our lower bound stands in contrast to many recent results
on simple and prior-independent auctions and fundamentally involves
the interplay between bidder competition, non-identical distributions,
and a very close (but still constant) approximation of the optimal
revenue. It effectively shows that the only way to achieve a
sufficiently good constant approximation of the optimal revenue is
through a detailed understanding of bidders' valuation
distributions. Our upper bound is constructive and applies in
particular to a variant of the empirical Myerson auction, the natural
auction that runs the revenue-maximizing auction with respect to the
empirical distributions of the samples. 

Our sample complexity lower bound depends on the set of allowable
distributions, and to capture this we introduce $\alpha$-strongly
regular distributions, which interpolate between the well-studied
classes of regular ($\alpha=0$) and MHR ($\alpha=1$) distributions.
We give evidence that this definition is of 
independent interest.
\end{abstract}

\section{Introduction}

Comparing the revenue of two different auctions requires an analysis
framework for
trading off performance on different inputs.  For instance, in a
single-item auction, a second-price auction with a reserve price~$r >
0$ will earn more revenue than a second-price auction with no
reserve price on some inputs, and less on others.  Which auction is
better?

The conventional approach in auction theory is Bayesian, or
average-case, analysis.  That is, bidders' valuations are assumed to be
drawn from a distribution, and one auction is defined to be better
than another if it has higher expected revenue with respect to this
distribution.  The optimal auction is then the one with the highest
expected revenue.  The optimal auction depends on the assumed
distribution, in some cases in a detailed way.

While there is now a significant body of work on worst-case revenue
maximization (see~\cite{HK07}), a majority of modern computer science
research on revenue-maximizing auctions uses Bayesian analysis to
measure auction performance (see~\cite{hbook}).
Since the comparison between auctions depends fundamentally on the
assumed distribution, an obvious question is: {\em where does this
  prior distribution come from, anyway?}

In most applications, and especially in computer science contexts, the
answer is equally obvious: {\em from past data}.  For example, in
Yahoo!'s keyword auctions, Bayesian analysis was used to provide
guidance on how to set per-click reserve prices, and the valuation
distributions used in this analysis are derived straightforwardly from
bid data from the recent past~\cite{OS09}.  This is a natural
approach, but how well does it work?

\subsection{The Model}

The goal of this paper is to understand how much data is necessary and
sufficient to guarantee near-optimal expected revenue.
Our model is the following.  There are $k$ bidders in a
single-item auction.  The valuation (i.e., willingness-to-pay) of
bidder $i$ is a sample from a distribution $\disti$.  The $\disti$'s
are independent but not necessarily identical.

The distribution $\dists = \dist_1 \times \cdots \times \dist_k$ is
unknown to the seller.  The ``data'' comes in the form of $m$ 
independent and identically distributed (i.i.d.)
samples $\samples$ from $\dists$ --- equivalently, $m$ i.i.d.\ samples
from each of the $k$ individual distributions
$\dist_1,\ldots,\dist_k$.  The seller observes the samples and then
commits to a truthful auction $\A$.\footnote{An auction is {\em
    truthful} if truthful                
  bidding is a dominant strategy for every bidder.  That is: for every          
  bidder~$i$, and all possible bids by the other bidders, $i$                   
  maximizes its expected utility (value minus price paid) by bidding            
  its true value.  
For single-item auctions, the optimal expected revenue of any
(possibly non-truthful) auction, measured  at 
a Bayes-Nash equilibrium with respect to the prior distribution, is no
larger than that of the optimal truthful auction.
Also, the restriction to dominant strategies is natural given our assumption
  of an unknown distribution. 
}
We call this function from samples to
auctions an {\em $m$-sample auction strategy}.
The seller then earns the revenue of its chosen auction $\A$ on
the ``real'' input, a fresh independent sample $\inp$ from
$\dists$.  See also Figure~\ref{f:multiple}.
We can state our main question as follows.
\begin{itemize}

\item [(*)] {\em How many samples $m$ are necessary and sufficient for
  the existence of an $m$-sample auction strategy that, for every
  distribution   $\dists$ in some class $\D$, has expected revenue at least
  $(1-\eps)$ times that of the optimal auction for $\dists$?}

\end{itemize}
The expected revenue of an auction strategy is with respect to both the samples
$\samples$ and the input $\inp$ --- i.e., over $m+1$ i.i.d.\ samples from
$\dists$.  The expected revenue of an optimal auction is with respect
to a single sample (the input) from $\dists$.  Our formalism is
inspired by computational learning theory~\cite{V84}.

\begin{figure}
\begin{center}
\includegraphics[scale=.7]{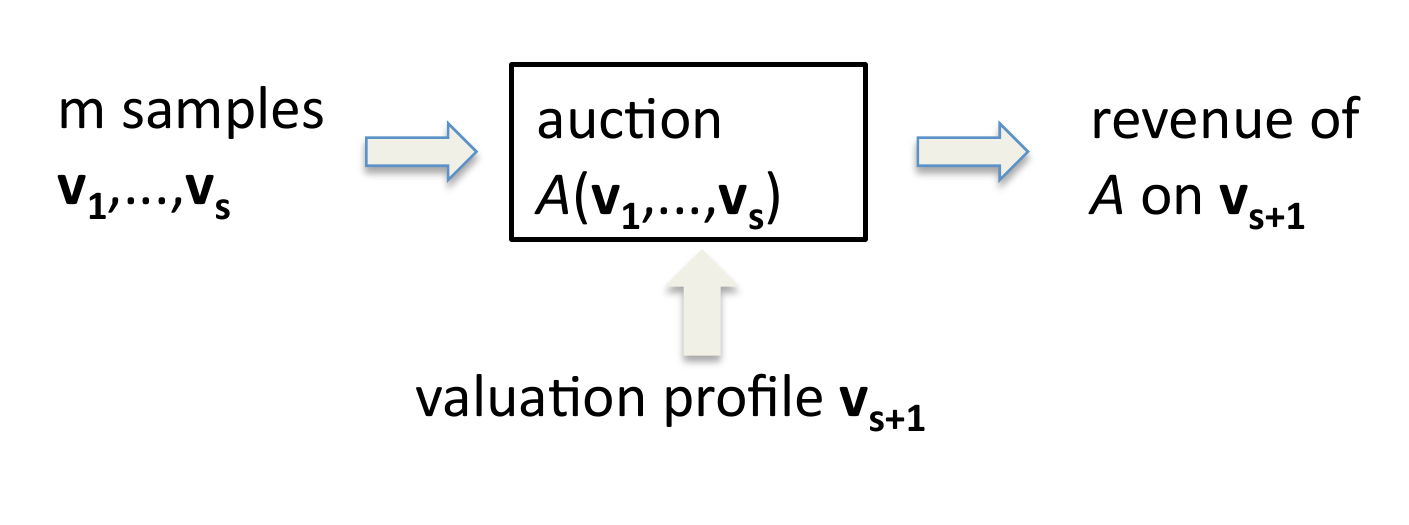}
\caption{A single-item auction $A$ is
  chosen as a function of $s$ i.i.d.\ samples $\vals_1,\ldots,\vals_s$
  from an
  unknown distribution $\dist_1 \times \cdots \times \dist_n$, and
  applied to a fresh sample
  $\vals_{s+1}$ from the same distribution.
The benchmark is the expected revenue of the Myerson-optimal auction
for $\dist_1,\ldots,\dist_n$.}
\label{f:multiple}
\end{center}
\end{figure}

The answer to the question~(*) could be a function of up to three
different parameters: the error tolerance $\eps$, the number~$k$ of
bidders, and the set $\D$ of allowable
distributions\footnote{As the distribution $\dists$ is unknown, we seek
uniform sample complexity bounds, meaning bounds that depend only on
$\D$ and not on $\dists$.}.
It is clear
that some restriction on $\D$ is necessary for the question~(*) to be
interesting: without any restriction, no finite number of samples is
sufficient to guarantee near-optimal revenue, even when there is only
one bidder.\footnote{To see this, consider all distributions that take
  on a value   $M^2$ with probability $\tfrac{1}{M}$ and 0 with
  probability $1-\tfrac{1}{M}$.  The optimal auction for such a
  distribution earns expected revenue $M$.  It is not difficult to
  prove that, for every $m$, there is no $m$-sample auction strategy
  with near-optimal revenue for every such distribution --- for
  sufficiently large $M$, all $m$ samples are~0 with high probability
  and the 
  auction strategy has to resort to an uneducated guess for $M$.}

Our model can be viewed as an interpolation between worst-case and
average-case analysis.  It inherits much of the robustness of the
worst-case model, since we demand guarantees for every underlying
$\dist$, while allowing very good approximation guarantees with
respect to a strong benchmark.

\subsection{Distributional Assumptions}

Two distributional assumptions that have been extensively used (see
e.g.~\cite{hbook}) are the {\em regularity} and {\em monotone hazard
  rate (MHR)} conditions.  The former asserts that the
``virtual valuation'' function $\vali -
\tfrac{1-\disti(\vali)}{\densi(\vali)}$ is nondecreasing, where
$\densi$ is the density of $\disti$, while the
second imposes the strictly stronger condition that
$\tfrac{\densi(\vali)}{1-\disti(\vali)}$ is nondecreasing.
The ``most tail-heavy'' regular distribution has the distribution
function $\disti(\vali) = 1 - \tfrac{1}{\vali+1}$, while the most
tail-heavy MHR distributions are the exponential distributions.

Our lower bound on the sample complexity of revenue maximization depends
on the set of allowable distributions, and to capture this we introduce a parameterized
condition that interpolates between the regularity and MHR conditions;
this condition is also useful in other contexts (see Section~\ref{s:strong}).
\begin{definition}{\sc ($\alpha$-Strongly Regular Distribution)}\label{d:a}
Let $\dist$ be a distribution with positive density function $\dens$
on its support $[a,b]$, where $0 \le a < \infty$ and $a \le b \le
\infty$.  Let $\vv(\val) = \val - \tfrac{1-\dist(\val)}{\dens(\val)}$
denote the corresponding virtual valuation function.
$\dist$ is {\em $\alpha$-strongly regular} if
\begin{equation}\label{eq:a}
\vv(y) - \vv(x) \ge \alpha(y-x)
\end{equation}
whenever $y > x \ge 0$.
\end{definition}
For distributions with a differentiable virtual valuation function
$\vv$, condition~\eqref{eq:a} is equivalent to $\tfrac{d\vv}{dv} \ge
\alpha$.
Regular and MHR distributions are precisely the 0- and 1-strongly regular
distributions, respectively.
A product distribution $\dists = \dist_1 \times \cdots \times \dist_k$
is called $\alpha$-strongly regular if each $\disti$ is
$\alpha$-strongly regular.
For the lower bound, we take the set $\D$ of allowable distributions
in~(*) to be the $\alpha$-strongly regular distributions for a
parameter $\alpha \in [0,1]$.

\subsection{Our Results}

Our main result is that 
$m = \poly(k,\tfrac{1}{\eps})$ samples are necessary and sufficient for the
existence of an $m$-sample auction strategy that, for every
strongly regular distribution $\dists$, has
expected revenue at least $(1-\eps)$ times that of an optimal auction.

Both our upper and lower bounds on the sample complexity of revenue
maximization are significant.
For the lower bound, it is far from obvious that the number of samples
per bidder needs to depend on $k$ at all, let alone polynomially.
Indeed, for many relaxations of the problem we study, the sample
complexity is a function of $\eps$ only.
\begin{itemize}

\item If there is an unlimited supply of items (digital goods), then
  the problem reduces to separate single-bidder problems, for which
  $\poly(\tfrac{1}{\eps})$ samples suffice for a
  $(1-\eps)$-approximation for all regular
  distributions~\cite[Lemma 4.1]{DRY10}.

\item If bidder valuations are independent and {\em identical} draws
  from an unknown regular distribution, then $\poly(\tfrac{1}{\eps})$
  samples suffice for a $(1-\eps)$-approximation~\cite[Theorem
  4.3]{DRY10}.

\item If only a $\tfrac{1}{4}$-approximation of the optimal expected
  revenue is required, then only a {\em single} sample is required.
  This follows from a generalization of the Bulow-Klemperer
  theorem~\cite{BK96} to non-i.i.d.\ bidders~\cite[Theorem 4.4]{HR09}.

\end{itemize}
Thus, the necessary dependence on $k$ fundamentally involves the
interplay between bidder competition, non-identical distributions, and
a very close (but still constant) approximation of the optimal revenue.

On a conceptual level, our lower bound shows that
designing $c$-approximate auctions for constants $c$ sufficiently
close to~1 is a qualitatively different problem than for more modest
constants like $\tfrac{1}{4}$.  For example, previous work has
demonstrated that auctions with reasonably good approximation factors
are possible with minimal dependence on the valuation distributions
(e.g.~\cite{CHK07,HR09,AKW14}) or even, when there is no bidder with a unique
valuation distribution, with {\em no} dependence on the valuation
distributions~\cite{D+11,DRY10,RTY12}.
Another interpretation of some previous results, such
as~\cite{C+10,CHK07}, is the existence of constant-factor approximate
auctions that derive no benefit from bidder competition.  Our lower
bound identifies, for the first time, a constant approximation
threshold beyond which ``robustness'' and ``prior-independence''
results of these types cannot extend.  Our argument
formalizes the idea that, with
two or more 
non-identical
bidders, the {\em only way} to achieve a sufficiently
good constant approximation of the optimal revenue is through a
detailed understanding of bidders' valuation distributions and an
essentially optimal resolution of bidder competition.

\hide{
Finally, a variant of our lower bound argument shows that there is an
$\eps > 0$ such that no finite number of samples guarantees a
$(1-\eps)$-approximation with respect to every regular ($\alpha=0$)
distribution.  Again, only one sample is required for a
$\tfrac{1}{2}$-approximation~\cite{BK96,HR09}.
}


We provide an upper bound on the number of samples needed for
near-optimal approximation by analyzing a natural auction strategy.
Recall that for a distribution $\dists$ that is known a priori,
Myerson's optimal auction gives the item to the bidder with the highest
virtual valuation $\vvi(\vali) = \vali -
\tfrac{1-\disti(\vali)}{\densi(\vali)}$, or to no one if all virtual
valuations are negative~\cite{M81}.
The {\em empirical Myerson auction} is the obvious analog when one has
data rather than distributional knowledge: define $\bar{\disti}$ as
the empirical distribution of the samples from $\disti$, and run the
optimal auction for $\bar{\dists}$.\footnote{Since the empirical
 distributions are generally not regular even when the underlying
 distributions $\dists$ are, a standard extra ``ironing'' step is
 required; see Section~\ref{s:prelim} for details.}
We prove that 
a variant on
the empirical Myerson auction has expected revenue at least $(1-\eps)$
times optimal provided it is given a sufficiently large polynomial
number of samples.\footnote{Left unmodified, the empirical Myerson
  mechanism can be led astray by poor approximations at the upper
  end of the valuation distributions caused by a small sample effect.
We prove that excluding the very highest samples from the empirical
distributions addresses the problem.}
A key aspect in our analysis is identifying the (non-pointwise) sense
in which empirical virtual valuation functions approximate the actual
virtual valuation functions; this is non-trivial even for the special
case of MHR distributions.

\hide{
Since no finite number of samples is sufficient for a
$(1-\eps)$-approximation for any auction
strategy when $\alpha = 0$ (as our lower bound shows), our analysis
must also control the degradation in this approximation as $\alpha
\downarrow 0$.
}

\subsection{Technical Approach}

The proofs of
our upper and lower bounds are fairly technical, so we
provide here an overview of the main ideas.  We begin with the upper
bound, which roughly consists of the following steps.
\begin{enumerate}


\item
(Lemma~\ref{lem:samp-acc})~
For some fixed bidder with distribution $\dist$, consider the
  corresponding $m$ samples $\val_1 \ge \cdots \ge \val_m$.
  Define the  ``empirical   quantile'' $\bar{q}_j$ of $\val_j$ as
  $\tfrac{2j-1}{2m}$, the expected quantile of the $j$th order
  statistic.  Taking a ``net'' of quantiles and applying standard large
  deviation bounds shows that all but the bottom $\hgxi$ fraction of the
  empirical quantiles are good multiplicative approximations of their
  expectations with high probability; here $\hgxi > 0$
  is a key parameter that will depend on $k$ and $\eps$.

\item (Lemma~\ref{lem:phi-q-bdd})~
\label{itm:add-factor}
Recall that the expected revenue of an auction equals its
  expected virtual welfare~\cite{M81}.  Myerson's optimal auction
  maximizes virtual welfare pointwise, whereas our auction maximizes
  (ironed) empirical virtual welfare pointwise.
In a perfect world, we would be able to argue that the empirical
  virtual valuation functions are good pointwise approximations of the
  true virtual valuation functions, and hence the expected virtual
  welfare of our auction is close to that of Myerson's auction.
  Unfortunately, good relative approximation of quantiles does not
  necessarily translate to good relative approximation of virtual
  valuations.  The reason is that a virtual valuation function $\val -
  \tfrac{1-\dist(v)}{\dens(v)}$ can change arbitrarily rapidly in a
  region where the density changes rapidly (even for MHR distributions).

We instead prove a different sense in which
empirical virtual values approximate actual virtual values, working in
  the (quantile) domain as well as in the range of the virtual valuation
  functions.
  Recall that the quantile $q(v)$ is defined as $1 -F(v)$.
  We show that for suitable $\Delta_1, \Del_2 > 0$,
  for all but the top $\hgxi$ fraction of quantiles in $[0,1]$, with
  high probability
  the empirical virtual value $\bar{\vv}(q)$ is sandwiched between
  $\vv(q(1+\Delta_1))$ and $\vv(q/(1+\Delta_2))$, modulo small additive 
 factors. 
  (By $\vv(q)$ we mean
  $\vv(F^{-1}(1-q))$.)
  The additive factors are functions of $1/q$, as well as of $k$ and $\eps$,
  which complicates the analysis.

\item (Lemmas~\ref{lem:cond-loss} and~\ref{lem:scnd-term})~
Consider a fixed bidder $i$.  By
  the previous step, up to additive factors, we can lower bound the virtual welfare
  contributed by a bidder $i$ with a quantile
  $q_i = 1-\disti(\vali)$ outside the top $\hgxi$ fraction in the
  empirical Myerson mechanism by the virtual value  
  contributed by $i$ in the optimal auction when it has a quantile of
  $q_i(1+\Delta_1)$.  Or not quite: an additional issue is that the
  empirical virtual valuation of a different bidder $j$ with quantile
  $q_j$ might be larger than its true virtual value, leading the
  empirical Myerson auction to allocate to $j$ over the rightful winner
  $i$, and resulting in a reduction in the total virtual welfare being accumulated
  as compared to the actual Myerson auction.  The difference in virtual values
  is bounded by the additive factors described in (\ref{itm:add-factor});
 as the outcomes are determined in the empirical auction,
  we will need two sets of factors:  for bidder $i$, the factors from the lower sandwiching bound,
  and for bidder $j$, the factors from the upper sandwiching bound.
  We will show that there is only a small 
  probability of large additive factors, which suffices to
  bound the expected reduction in revenue when the additive factors are large.
  When the additive factors are small, their contribution to the reduction in revenue is also small.
All the reductions end up being polynomial functions of $k$ and $\eps$.

\item (Lemma~\ref{lem:frth-term-i} and~\ref{lem:fifth-term})~
  There is one more issue. Because of the shift in quantile space --- we  compare the virtual value
  in the Myerson auction at quantile $q_i(1+\Delta_1)$ to the virtual value in the empirical auction at quantile $\qi$ ---
  and also because the reserve prices in the two auctions may differ, we also have to analyze the revenue loss at
  the lower end of the distributions, or more precisely, around the reserve prices. This too is a polynomial function of $k$ and $\eps$.
\end{enumerate}

We now discuss the lower bound proof.
This involves arguing that, if the number of samples is too small,
then for every auction strategy, there exists a distribution for which
the auction strategy's expected revenue is not near-optimal.
We prove this by exhibiting a ``distribution of distributions'' and
proving that every auction strategy has expected revenue --- where the
expectation is now with respect to both the initial random choice of
the valuation distributions, and then with respect to both the $m$
samples and the input --- bounded away from the expected revenue of an
optimal auction (where the expectation is over both the choice of
distributions and the input).
We are unaware of any other lower bounds in auction theory that have
this form.

Our construction involves taking a base set of ``worst-case''
$\alpha$-strongly regular distributions and truncating them at
random points.  A key observation is that, when such a distribution is
truncated at a point $H_i$, the corresponding virtual valuation function
is linear with coefficient $\alpha$ except at the truncation point,
where the virtual valuation jumps to $H_i$.  The high-level intuition
is that, when confronted with valuations that are higher than those
seen in any of the samples, no auction can know whether a high
valuation $\val$ corresponds to a truncation point (with virtual value
$\val$) or not (with virtual value only $\alpha (v - 1)$).  Properly
implemented, this idea can be used to prove that every auction
strategy errs with constant probability on precisely the set of inputs
that contribute the lion's share of the optimal revenue.  The lower
bound follows.

\subsection{Prior and Concurrent Related Work}

We provide detailed comparisons only with 
the papers most closely related to the present work.
For previously studied models about revenue-maximization with an unknown
  distribution, which differ in various respects from our model,
see \cite{BBDSics11,KL2003focs}.
For other uses of samples in auction design that differ from ours,
see Fu et al.~\cite{FHHK14}, who use samples to extend the 
Cr\'emer-McLean theorem~\cite{CM85}
to partially  known valuation distributions, and Chawla et
al.~\cite{CHN14},  who design non-truthful
auctions that both have equilibrium
revenue within a constant of optimal
and enable accurate inference about the valuation distribution
from samples.
For other ways to parameterize partial knowledge about valuations,
see e.g.~Azar et al.~\cite{ADMW13} and Chiesa et al.~\cite{CMZ12}.
For asymptotic optimality results in various symmetric 
settings, which identify conditions under which the expected revenue
of some auction of interest (e.g., second-price) approaches the
optimal with an increasing number of (i.i.d.) bidders, see
Neeman~\cite{N03},
Segal~\cite{segal},
Baliga and Vohra~\cite{BV03},
and Goldberg et al.~\cite{G+06}.
For applications of learning theory
concepts to prior-free auction design in unlimited-supply settings,
see Balcan et al.~\cite{BBHM08}.
Finally, the technical issue of ironing from samples comes up also in 
Ha and Hartline \cite{Ha13}, in the context of prior-free mechanism
design, and the aforementioned Chawla et al. \cite{CHN14}.
Our goal of obtaining a $(1-\eps)$-approximation of the maximum revenue
achieved by any auction is impossible in the more demanding
settings of~\cite{Ha13,CHN14}.

Elkind~\cite{elkind2007} studies a learning problem closely related to
ours, in the restricted setting of discrete distributions with known
finite supports but with unknown probabilities.  In the model
in~\cite{elkind2007}, learning is done using an oracle that compares
the expected revenue of pairs of auctions, and $O(n^2K^2)$ oracles
calls suffice to determine the optimal auction (where $n$ is the
number of bidders and $K$ is the support size of the distributions).
Elkind~\cite{elkind2007} notes that such oracle calls can be
implemented approximately by sampling (with high probability), but no
specific sample complexity bounds are stated.

Dhangwatnotai et al.~\cite{DRY10}, motivated by ``prior-independent''
auctions that are simultaneously approximately optimal for a wide
range of valuation distributions, implicitly studied the single-bidder
version of the learning
problem we that study.  With one bidder, the goal is to learn approximately
the monopoly price of an unknown distribution from samples.  Their results
imply sample complexity upper bounds for this problem of
$O(\eps^{-2})$ and $O(\eps^{-3})$ for MHR and regular distributions,
respectively.\footnote{Similarly, our sample complexity upper bound
  naturally 
leads to a prior-independent single-item auction.  This auction
achieves a $(1-\eps)$-approximation of the optimal auction when
bidders' valuations are drawn from different regular
distributions $F_1,\ldots,F_k$ and there are sufficiently many bidders
of each type.}
As our results show, the learning problem is quite different and more
delicate with multiple non-i.i.d.\ bidders.

The papers of Cesa-Bianchi et al.~\cite{CGM15} and Medina and
Mohri~\cite{MM14} give algorithms for learning the optimal
reserve-price-based single-item auction.  
Our problem of learning the best single-item auction --- whether
reserve-based or otherwise --- is harder.  With non-i.i.d.\ bidders,
there need not be a reserve price for which the second-price auction
has expected revenue more than~50\% times the optimal~\cite{HR09}.

Concurrently with our work,
Dughmi et al.~\cite{dughmi2014sampling} proved negative
results (exponential sample complexity) for learning near-optimal
mechanisms in  multi-parameter settings that are much more complex
than the single-item auctions studied here.  The paper also contains
positive results for restricted classes of mechanisms.

\subsection{Subsequent Related Work}

The preliminary version of this paper~\cite{CR14} motivated several
follow-up works.  Huang et al.~\cite{HMR15} study the single-bidder
version of our problem, studied implicitly in~\cite{DRY10},
and give optimal sample complexity bounds 
under several different
distributional assumptions.  Both the upper and lower bounds
in~\cite{HMR15} improve, in terms of the dependence on $\eps^{-1}$,
over those implied by the present work.  This is also the only paper other
than the present work that proves any sample
complexity lower bounds.

Morgenstern and Roughgarden \cite{MR15} adapt tools from statistical
learning theory~\cite{AB99} to give general sample complexity upper bounds that
cover all single-parameter settings for bounded or MHR valuation
distributions.  The results in~\cite{MR15} apply to many more
environments than single-item auctions, and improve over the sample
complexity upper bound of the present work (even for single-item
auctions), but unlike the present work, their results do not apply to
regular distributions and do not result in computationally efficient
learning algorithms.

Very recently, Devanur et al.~\cite{DHP15} devised a different
learning algorithm for our problem which provides a strict improvement on our
upper bound.  The sample complexity bound in~\cite{DHP15} is roughly the
same as in~\cite{MR15}, but the learning algorithm is computationally
efficient and also accommodates regular distributions.

Finally, Cole and Shravas investigated further applications of $\alph$-strongly regular distributions~\cite{CR15}.

\section{Preliminaries}\label{s:prelim}


This section reviews Myerson's optimal single-item auction~\cite{M81}
for the case of known distributions.
There are $k$ bidders, and for each bidder~$i$
there is a distribution $\Fi$ from which its
valuation is drawn.\footnote{Our results extend to the case
of $k$ groups of an arbitrary number of bidders, where all the bidders
from group $i$ have i.i.d.\ valuations drawn from $\disti$.
Then the sample bounds are a function of the number $n$ of bidders,
rather than the number $k$ of groups,
but the bound is on the number of samples needed from each group of bidders.}

For each buyer~$i$, the auctioneer computes a virtual valuation
$\vvi(v) = v - [1 - \Fi(v)]/ \efi(v)$, where $\efi$ is the density function corresponding to $\Fi$.
$\vvi(v)$ is required to be a non-decreasing function of $v$.
This holds by definition for regular distributions; in general,
if this does not hold $\vvi$ can be modified (or \emph{ironed})
so that it does hold, as implicitly explained in the next paragraph.
Next, the auctioneer runs an analog of a second-price
auction on the virtual values of the bids (virtual bids for short):
the bidder, if any, with the highest non-negative virtual bid wins the auction
(ties are broken arbitrarily)
and is charged the minimum bid needed to win
(or at least to tie for winning).
More precisely, let $i$ be the winning bidder and let $b_2$ be the
second highest virtual bid.
Then the price is $\vvi^{-1}(\min\{0,b_2\})$. We note that
$\vvi^{-1}(0)$ can be viewed as a bidder-specific reserve
price for $i$; it is also called the \emph{monopoly price} for $i$.

We can also describe the auction in terms of a revenue function.
This also allows for situations where $\vvi(v)$ is not a nondecreasing
function of $v$.
The revenue function is computed in quantile space:
$\qi(v) = 1 -\Fi(v)$ is the probability that $i$ will have a valuation
of at least $v$.  
Now we view $v$ as a function of $\qi$.
We introduce the expected revenue function, $\Ri(\qi)$. It is a function of the quantile $\qi$:
$\Ri(\qi) = v(\qi) \cdot \qi$ is the expected revenue if $i$ is the
sole bidder and $v(\qi)$ is the price being charged.
The auctioneer computes the smallest concave upper bound $\CRi(q)$ of
$\Ri(q)$. 
Now $\vvi(v(\qi))$ is defined to be
the slope of $\CRi(\qi)$ (this yields an increasing function $\vvi$,
which coincides with the previous definition in the case of regular
distributions).
At points where there is no unique slope we choose $\vvi(v(q)) = \lim_{(q'>q) \ra q} \vvi(v(q'))$.
The auction then proceeds as before.
Henceforth, overloading notation, we write $\vvi(\qi)$ rather than $\vvi(v(\qi))$.

Myerson~\cite{M81} proved that for
every auction, the expected virtual welfare equals the expected
revenue.  This result is important because, in many situations, it is much
easier to reason about expected virtual welfare than directly about
expected revenue.
\begin{theorem}[Myerson]
\label{thm:Myerson}
The expected revenue of any single-item auction is given by
\[
\sum_{i=1}^k E_{\qi}[ \poti(\qi) \cdot x_i(\qi)],
\]
where $x_i(\qi)$ is the probability (over others' valuations and any
coin flips by the auction) that $i$ wins the item 
with a bid at quantile $\qi$ in $\Fi$.
\end{theorem}
Let $\q = (q_1, q_2, \ldots, q_k)$ be a vector of quantiles drawn from $F_1 \times F_2 \times \ldots \times F_k$.
We can rewrite the expected revenue as
\begin{align}
\label{eqn:alloc-to-top-vv}
\sum_{i=1}^k E_{\qi}[ \poti(\qi) \cdot x_i(\qi)],
= \sum_i \int_{\q} \poti(\qi) I_i(\q) ~d\q,
\end{align}
where $I_i(\q)$ is the indicator function showing whether $i$ wins when the bids are at quantiles $\q$
(or more generally, the probability that it wins).
This immediately implies that allocating to a bidder with the highest virtual value, i.e.\ Myerson's auction,
is optimal.

\section{Statement of Main Results}
\label{sec:results}

We formally state our upper and lower bound results in turn.

\begin{theorem}
\label{thm:emp-myer-perf}
In a single-item auction with $k$ bidders with independent regular
valuation distributions, if
$m = \Omega( \frac {k^{10}}{\eps^7} \ln^3 \frac{k} {\eps} )$,
then there is an $m$-sample auction strategy with  expected revenue 
at least $1-\eps$ times that of an optimal auction.
\end{theorem}
The auction strategy in Theorem~\ref{thm:emp-myer-perf} is a variant
on the ``empirical Myerson auction,'' described in detail in
Section~\ref{ss:ema}.

Our lower bound result has an analogous form, 
although the polynomial in $k$ and $\eps$ is considerably smaller.
Our lower bound grows larger as $\alph \ge 0$ grows smaller.

\begin{theorem}
\label{thm:lower}
For every auction strategy $\Sigma$, for every $k \ge 2$,
for every sufficiently small $\eps > 0$,
for every $\alph \ge 0$ and $m$ satisfying:
\begin{enumerate}[i.]
\item
$\alph  = 1$ and
$m \le \left(\frac {1 - \ln 2} {96 e^3  \min\{1, \frac ke \} \ln \max\{e, k\}}\right)^{1/2} \frac  k{\sqrt {\eps} } $;
\item
$0< \alph < 1$, $\alph^{1/(1 - \alph)} \ge \tfrac 1k$, and
$m \le \left( \frac {1 - \alph2 ^{1 - \alph} }{96 e^3}\right)^{1/(1 + \alph)} \frac{k}{\eps^{1/(1 + \alph)} }$;
\item
$0< \alph < 1$, $\tfrac {1}{2m} < \alph^{1/(1 - \alph)} < \tfrac 1k$, and
$m \le  \left( \frac {1 - \alph2 ^{1 - \alph} }{96 e^3}\right)^{1/(1 +
  \alph)} \left(\frac{1}{k \alph^{1/(1 - \alph)}}\right)^{\alph/(1 +
  \alph)} \frac{k}{\eps^{1/(1 + \alph)} }$;
\item
$0< \alph < 1$, $\alph^{1/(1 - \alph)} \le \tfrac {1}{2m}$, and
$m \le \frac {(1 - \alph 2^{1 - \alph}) 2 ^{\alph}}{96e^3}\frac
  {k}{\eps}$; 
\item
$\alph = 0$ and
$ m \le \frac {1}{96e^3}\frac {k}{\eps}$,
\end{enumerate}
there exists a set $\dist_1,\ldots,\dist_k$ of $\alpha$-strongly regular
valuation distributions such that the expected revenue of $\Sigma$
(over the $m$ samples and the input) 
is less than $1-\eps$ times that of an optimal auction 
for $\dist_1,\ldots,\dist_k$.
\end{theorem}
At the two extreme points of MHR distributions ($\alpha=1$) and
regular distributions ($\alpha=0$), the lower bound in
Theorem~\ref{thm:lower}
is $\Omega(\tfrac{k}{\sqrt{\eps \ln k}})$ and
$\Omega(\tfrac{k}{\eps})$, respectively.  In all cases, the dependence
on the number of bidders is linear or near-linear.

We prove these two theorems in Sections~\ref{sec:upper}
and~\ref{sec:lower} respectively. 
Before that, in the next section, we briefly indicate two applications
of $\alph$-strong regularity. 

\section{Applications of Strong Regularity}
\label{s:strong}

We believe our definition of $\alpha$-strongly regular distributions
is of independent interest. Almost all previous expected revenue
approximation guarantees for auctions in Bayesian settings
apply to one of three sets of valuation distributions: all
distributions, all 
regular distributions, or all MHR distributions (see
e.g.~\cite{hbook}).  Strongly regular distributions interpolate
between regular and MHR distributions, and should broaden the reach of
many existing approximation bounds that are stated only for MHR
distributions.
To prove this point, we mention a couple of examples of such extensions;
we are confident that many others are possible, as has been shown subsequently by Cole and Shravas~\cite{CR15}.

The following property of MHR distributions is well known~\cite[Lemma
  4.1]{HMS08}. 
\begin{lemma}[\cite{HMS08}]
\label{lem:quant-bdd1}
Let $F$ be an MHR distribution with monopoly price $r$.
If $q(r)$ is the quantile of valuation $r$ in the distribution $F$,
then $q(r) \ge \frac 1e$.
\end{lemma}
We next show how to generalize this result to $\alpha$-strongly regular
distributions.
\begin{lemma}
\label{lem:quant-bdd2}
Let $F$ be an $\alpha$-strongly regular distribution with 
$\alpha \in (0,1)$ and monopoly price $r$.
If $q(r)$ is the quantile of valuation $r$ in the distribution $F$,
then $q(r) \ge \alph^{1/(1-\alph)}$.
\end{lemma}

\begin{proof}
Set $\lmbd = 1 -\alph$.
Let $h(\cdot)$ denote the hazard rate of $F$, 
and choose
$c$ so that  $h(r) = \frac {1}{\lmbd r +c}$.
Recall that $\pot(v) = v - \frac{1}{h(v)}$.
Since $\pot(r) = 0$ and hence $h(r) = 1/r$, we have $c = r(1 -
\lmbd)$.

The $\alph$-strong-regularity condition, $\dbydv \pot \ge \alph$,
implies that $1 + \frac {1}{h^2} \dbydv h \ge \alph$,
or $ \dbydv {} \left( \frac {1}{h} \right) \le \lambda$.
It follows that,
for all $v \le r$,
 $h(v) \le \frac {1}{\lmbd v +c}$
and hence
%
\[
h(v) \le \frac {1} {\lmbd (v - r) + r}.
\]

Now write $H(x) = \int_0^x h(v) dv$;
it is well known and easy to verify that $q(v) = e^{-H(v)}$.
We complete the proof by deriving
\begin{align*}
q(r) ~=~ e^{-H(r)} & =  e^{-\int_0^r h(v) dv} \\
& \ge  e^{-\left[\frac {1}{\lmbd} \log(r + \lmbd(v-r)) \, |_0^r \, \right]}\\
&= e^{-\frac {1}{\lmbd} \log\frac {r}{r(1-\lmbd)}} ~ = ~ e^{\log(1 - \lmbd)^{1/\lmbd}}
 =  (1 - \lmbd)^{1/\lmbd} ~=~ \alph^{1/(1 - \alph)}.
\end{align*}
\end{proof}

Hartline et al.~\cite[Theorem 4.2]{HMS08} study a revenue maximization
problem in social networks, and give a mechanism with approximation
guarantee 
\begin{equation}\label{eq:hms}
\frac{1}{4-\tfrac{2}{e}} \approx .306
\end{equation} 
when players' private valuations are drawn from
MHR distributions.\footnote{They also give a
  $\tfrac{1}{4}$-approximation algorithm for arbitrary valuation
  distributions.}
The MHR assumption is used only in applying Lemma~\ref{lem:quant-bdd1}.
Relaxing the distributional assumption to $\alpha$-strong regularity
and reoptimizing the proof in~\cite[Theorem 4.2]{HMS08} 
using Lemma~\ref{lem:quant-bdd2} extends this approximation
guarantee accordingly, with the term $1/e$ in~\eqref{eq:hms}
replaced by $\alph^{1/(1- \alph)}$.

For a second example, Hartline and Roughgarden~\cite[Theorem
  3.2]{HR09} consider downward-closed single-parameter
environments\footnote{A (binary) single-parameter environment is
  specified by a set of bidders and the feasible subsets of bidders
  that can 
  simultaneously win.  For example, if each bidder~$i$ wants a known
  bundle~$S_i$ of items, then the feasible subsets are those in which
  the bundles of the chosen bidders are pairwise disjoint.  Such an
  environment is downward closed if every subset of a feasible set is
  again feasible.}
and prove that, when bidders' valuations are drawn from MHR
distributions, the VCG 
mechanism with ``eager'' monopoly reserve prices\footnote{In more
  detail, one reserve price~$r_i$
  per bidder~$i$ is fixed in advance.  A bid is collected from each
  bidder.  Bidders who bid below their reserve prices are removed from
  further consideration.  From the remaining bidders, the mechanism
  chooses winners to maximize the sum of their bids, subject to
  feasibility.  The mechanism charges the unique prices for which
  losing bidders pay~0 and truthful bidding is a dominant strategy.
  See~\cite{HR09} for details.}
has expected revenue
at least 
$\tfrac{1}{2}$ times that of an optimal mechanism.
For an $\alpha$-strongly regular distribution~$F$ with monopoly price
$r$, and $v \ge r$, we have $\pot(v) \ge \alpha(v-r)$ and hence
$$
r + \frac{1}{\alpha}\pot(v) \ge v;
$$
this inequality generalizes Lemma~3.1 in~\cite{HR09}.
Following the proof in~\cite[Theorem  3.2]{HR09} shows that, for every
downward-closed single-parameter environment with
bidders valuations drawn from $\alpha$-strongly regular
distributions, the VCG mechanism with eager monopoly reserves has
expected revenue at least $\tfrac{\alpha}{\alpha+1}$ times that of an
optimal mechanism.

\section{The Lower Bound: Proof of Theorem~\ref{thm:lower}}\label{sec:lower}

\paragraph{Formal Statement}
Fix $\alpha \ge 0$ and $0 < \del \le 1$, where $\del$ is sufficiently small.
We show that for every auction strategy $\Sigma$, there
exists a set $\dist_1,\ldots,\dist_k$ of $\alpha$-strongly regular
distributions such that the expected revenue of the auction strategy
(over the samples and the input) is at most the following fraction of the
expected revenue of the optimal auction for $\dist_1,\ldots,\dist_k$:
\begin{align*}
1 -\eps(1, \del)  & =  1 - \frac {1 - \ln 2} {96 e^3  \min\{1, \frac ke \} \ln \max\{e, k\}
} 
  \del^2  && \text{for}~ \alph = 1 \\
%
1 -\eps(\alpha, \del)  & =    1 - \frac {1 - \alph 2^{1 - \alph}} {96 e^3} \del^{1 + \alph} && \text{for}~ \alph < 1
~\text{and}~ \tfrac 1k \le \alph^{1/(1-\alph)}  \\
%
1 -\eps(\alpha, \del)  & =    1 - \frac {1 - \alph 2^{1 - \alph}} {96 e^3} \del^{1 + \alph} \frac{1} {( k \alph^{1/(1 - \alph)})^{\alph}} && \text{for}~ \alph < 1
~\text{and}~ \tfrac {\del}{2k} < \alph^{1/(1-\alph)} < \tfrac 1k  \\
%
1 -\eps(\alpha, \del)  & =    1 - \frac {1 - \alph 2^{1 - \alph}} {96 e^3} 2^{\alph} \del  && \text{for}~ \alph < 1
~\text{and}~ \alph^{1/(1-\alph)} \le \tfrac {\del}{2k}  \\
%
1 -\eps(\alpha, \del)  & =    1 - \frac {1} {96 e^3} \del && \text{for}~ \alph = 0
\end{align*}

We note that if $\alph < 1$, then $\alph 2^{1 - \alph} < 1$ also.
In addition, for fixed $k$, $\lim_{\alph \ra 0} [k \alph^{1/(1 - \alph)}]^{\alph} = 1$.
Substituting $k/m$ for $\del$ yields the bounds in
Theorem~\ref{thm:lower}.
For
$\tfrac 1k \le \alpha < 1$ and sufficiently small constant $\eps > 0$,
$\Omega(k/\eps^{1/(1 + \alph)} )$ samples
are necessary for a $(1-\eps)$-approximation.
For the MHR ($\alpha=1$) case,
$\Omega(k/\sqrt{\eps \ln k})$ samples are necessary,
and for the regular ($\alph =0$) case,
$\Omega(k/\eps)$ samples are needed.

\paragraph{The Base Distributions}
We identify the worst-case distributions for a given $\alpha \ge 0$.
Specifically, for $v \in [0,\infty)$, consider
\begin{align*}
\dist^{\alpha}(v) & =  1 - \left(\frac{1} {1 + (1-\alpha)v}
\right)^{\tfrac{1}{1-\alpha}} && \text{for}~0 \le \alph < 1; \\
F^1(v) & =  1 - e^{-v} && \text{for}~\alph = 1;\\
\dens^{\alpha}(v) &=  \left(\frac {1} {1 +
(1-\alpha) v}
\right)^{\tfrac{2-\alpha}{1-\alpha}} && \text{for}~0
\le \alph < 1; \\
f^1(v)& =  e^{-v}, && \text{for}~\alph = 1.
\end{align*}
The corresponding hazard rates are
\begin{align*}
h^{\alpha}(v) & =  \frac{1}{1 + (1-\alpha)v} && \text{for}~0 \le \alph < 1;\\
h^1(v)& =  1 && \text{for}~\alph = 1;
\end{align*}
with virtual valuation
$$
\vv^{\alpha}(v) = \alpha v-1~~~\text{for}~0 \le \alph \le 1.$$
A quick calculation shows that
\begin{align}
\label{eqn:F-inv}
(\F^{\alph})^{-1} (q) = \left\{\begin{array}{ll}
                                                         \frac {1} {1 - \alph} \left[\left( \frac 1q \right)^{1 - \alph} -1 \right] & \text{if}~ 0 \le \alph < 1 \\
                                                         \ln \frac 1q & \text{if}~ \alph = 1.
                                             \end{array}
                                    \right.
\end{align}

\paragraph{The Construction}
We define a distribution over distributions.
Each bidder $i$ is either type A or type B (50/50 and independently).
For a type B bidder $i$, we draw $q$ uniformly from the interval
$[\tfrac{\del}{2k},\tfrac{\del}{k}]$ and set $H_i = (F^{\alpha})^{-1}(1-q)$.
We then define bidder $i$'s distribution $\disti$ as equal to
$\dist^{\alpha}$ on $[0,H_i)$ with a point mass with the remaining
probability   $1-\dist^{\alpha}(H_i)$ at $H_i$.
For a type A bidder we proceed similarly except that $H_i$ is set to $(\dist^{\alph})^{-1}(1 - \tfrac {\del}{2k})$.
These distributions are always $\alpha$-strongly regular.
An important point is that the virtual valuation of these bidders is
given by
\begin{equation}\label{eq:truncated}
\vv(v) =
\left\{
\begin{array}{cl}
\alpha v-1 & \mbox{if $v < H_i$}\\
H_i & \mbox{if $v = H_i$}.
\end{array}
\right.
\end{equation}
Ultimately, it is the gap in virtual valuation between these two cases
that is responsible for the lower bound in Theorem~\ref{thm:lower}.

\paragraph{A Preliminary Lemma}
Let $\qalph^A$ denote the monopoly price in an auction with a single type A bidder.
We define $q_0^A = \frac{\del} {2k}$, as this is the largest quantile
$q$ for which $\vv(q) \ge 0$ when $\alph = 0$.
For $\alph > 0$ we begin by determining the monopoly price $\qalph$ for $\dist^{\alph}$.
We note that $\vv(v) = 0$ when $v=\tfrac {1}{\alph}$.
From~\eqref{eqn:F-inv},
we deduce this occurs at $q_{\alph} = \alph^{1/(1 - \alph)}$ for $0 < \alph < 1$,
and at $q_1 = \tfrac 1e$ for $\alph = 1$.
Thus, for bidder $i$, for $0 < \alph < 1$, $\qalph^A = \max\{\tfrac{\del}{2k}, \alph^{1/(1 - \alph)} \}$,
and for $\alph=1$, $q_1^A = \max\{\tfrac{\del}{2k}, \tfrac {1}{e}\} = \tfrac {1}{e}$
(as we can assume that $k \ge 2$, $\del \le 1$).

Let $\dist^{A,\alph}$ denote the distribution of a type A bidder.
Now let $v^* = (\dist^{A,\alpha})^{-1}(\max \{ 1 - \qalph^A,  \tfrac{k-1}{k} \} )$, the value corresponding to
quantile $\min\{ \qalph^A, \tfrac{1}{k} \}$ in $\dist^{A,\alpha}$,
and let $R^* = \min\{k \qalph^A, 1\} \cdot v^*$,
$k$ times the revenue at this quantile.
From \eqref{eqn:F-inv}, we obtain that
\begin{equation}\label{eq:Rstar}
v^* = \left\{
\begin{array}{cl}

\frac {1} {1 - \alph} \left[ \max \left\{ \frac {1} {\qalph^A}, k \right\}^{1-\alph} -1 \right]  & \mbox{if $0 \le \alph < 1$}\\
\ln \max \left\{ \frac {1} {\qalph^A}, k \right\} = \ln \max \{e, k \}  & \mbox{if $\alph = 1$.}
\end{array}
\right.
\end{equation}
\begin{lemma}
{\sc (Upper Bound on Optimal Revenue)}.
\label{lem:rev-upp-bdd}
The expected revenue (over $\vals$) of the optimal auction (with
respect to the $H_i$'s) is at most $R^*$.
\end{lemma}
\begin{proof}
First, the expected
revenue of the optimal auction is upper bounded by that of the optimal
auction for the case where all $H_i$'s are $(\dist^{\alph})^{-1}(1 - \tfrac{\del}{2k})$ --- i.e., where
$\disti = \dist^{A,\alpha}$ for every $i$.  
Second, by symmetry, when bidders valuations are i.i.d.\ draws
from $\dist^{A,\alpha}$, every bidder has the same
purchase probability $q$
in the (symmetric) optimal auction,
and since there is only one item, this purchase probability $q$ is at most
$\tfrac{1}{k}$; it is also at most $\qalph^A$.
Third, we obtain an upper bound by dropping the constraint of selling
only one item and instead
optimally selling to each bidder with probability $q$.
Fourth, this is precisely $k$ times the revenue of selling to a
single bidder with valuation from $\dist^{A,\alpha}$ using the posted
price $(\dist^{A,\alpha})^{-1}(1-q)$.
Fifth, by regularity, selling to a single bidder with posted price
$(\dist^{A,\alpha})^{-1}( 1-q )$ with $q \le \min\{  \qalph^A, \tfrac{1}{k}  \}$ is no
better than selling with posted price
$v^* = (\dist^{A,\alpha})^{-1}(\max \{ 1 - \qalph^A,  \tfrac{k-1}{k} \} )$.
The expected revenue from any one bidder is therefore at most
the sale probability times $v^*$, namely $\min\{\qalph^A, \tfrac{1}{k}\} \cdot v^*$.
The overall revenue, with $k$ bidders, is thus at most
$k \times \min\{\qalph^A, \tfrac{1}{k}\} \cdot v^* = R^*$, as claimed.
\end{proof}

\paragraph{Overview of Proof}
The high-level plan is the following.
Fix an arbitrary auction strategy.
Think of the random choices as occurring in three stages: in the first stage, the
$\disti$'s are chosen; in the second stage, $m$ sample valuation
profiles $\vals^{(1)},\ldots,\vals^{(m)}$ are chosen (i.i.d.\ from
$\dist_1 \times \cdots \times \dist_k$); in the final
stage, the input $\vals$ is chosen (independently from
$\dist_1 \times \cdots \times \dist_k$).
We prove that the expected revenue of the auction strategy (with
respect to all
three stages of randomness) is at most $1 - \eps(\alpha,\del)$ times that of
the optimal auction (with respect to  all three stages or,
equivalently, the 
first and third stages only).\footnote{We prove this statement about
 the expected virtual welfare, which is equivalent by Theorem~\ref{thm:Myerson}.}
Again, $1 - \eps(\alpha,\del) < 1$ will be
independent of $k$.  This implies that, for every auction strategy,
there exists a choice of $\dist_1,\ldots,\dist_k$ such that the
expected revenue of the auction strategy is at most $1 - \eps(\alpha,\del)$
times the expected revenue of the optimal auction for the
distributions $\dist_1,\ldots,\dist_k$.

By Lemma~\ref{lem:rev-upp-bdd},
$R^*$ is an upper bound on the optimal auction's expected revenue
(equivalently, expected virtual welfare) for
every choice of $\dist_1,\ldots,\dist_k$.
The main argument is the following:
there is an event $\E$ such that, for every auction strategy:
\begin{itemize}

\item [(i)] the probability of $\E$ (over all three stages of
  randomness) is lower bounded by a function
  $\gamma(\del)$ of $\del$  (and independent of $k$ and $\alpha$);


\item [(ii)] given $\E$, the expected virtual welfare of the auction
  strategy is at least $\eps (\alpha,\del)R^*$ smaller than that of the optimal auction,
  where $\eps (\alpha,\del) > 0$ is a function of $\alpha$ and $\del$ only.

\end{itemize}
Since by \eqref{eqn:alloc-to-top-vv}, for each set of bids, the virtual welfare earned by the optimal auction is always at least
that of the auction strategy,
(i)--(ii) imply that
the expected virtual welfare (and hence revenue) of the optimal auction
exceeds that of the auction strategy by $\eps(\alph, \del)R^*$ 
for some constant $\eps(\alph, \del)$
depending on $\alpha$ and $\del$.  
Lemma~\ref{lem:rev-upp-bdd} then implies that
the auction strategy's expected revenue is at most $1 -
\eps(\alpha,\del)$ times optimal.

\paragraph{The Main Argument}
To define the event $\E$,
we use the principle of deferred decisions.  We can flip the second-
and third-stage coins before those of the first stage by sampling
quantiles --- $(m+1)k$ i.i.d.\ draws $\{ q_i^{(j)} \}$ from the uniform
distribution on [0,1].  (Once the distributions are chosen in the
first stage, the valuation $\val^{(j)}_i$ is just
$\dist_i^{-1}(1-q_i^{(j)})$.)
We further break the first-stage coin flips into two substages; in the
first, we determine bidder types (A and B); in the second, we choose
$H_i$'s for the type-B bidders.
The event $\E$ is defined as the set of coin flips (across all stages)
that meet the following criteria:
\begin{itemize}

\item [(P1)] There are exactly two quantiles of the form $q^{(m+1)}_i$
  that are at most $\tfrac{\del}{k}$, say of bidders $j$ and $\ell$;

\item [(P2)] $q^{(m+1)}_j$ and $q^{(m+1)}_{\ell}$ are greater than
  $\frac{\del} {2k}$;

\item [(P3)] for $i=1,2,\ldots,m$, $q^{(i)}_j$ and $q^{(i)}_{\ell}$
  are greater than $\tfrac{\del}{k}$;

\item [(P4)] one of the bidders $j,\ell$ is type A, the other is type
  B (we leave random which is which);

\item [(P5)] the type B bidder (from among $j,\ell$) has valuation equal
  to the maximum valuation from its distribution.

\end{itemize}

The next lemma corresponds to step~(i) in the proof approach above.
\begin{lemma}
\label{lem:of-of-i}
The probability of $\E$ (over all three stages of
 randomness) is lower bounded by 
$\frac{\del^2}{32e^3}$.
\end{lemma}
\begin{proof}
We first sample the $k$ quantiles
corresponding to the third stage.  Elementary computations show
that property (P1) holds with probability at least $\tfrac{1}{2e} \del^2 $ (independent of $\alpha$ and
$k$).
Conditioned on (P1) holding, (P2) holds with probability
$\tfrac 14$.
(P3) is independent of the first two properties,
as it depends only on the second-stage randomness, and it 
holds with constant probability of at least $\tfrac {1}{e^2}$
(independent of $\alpha,k$). 
Proceeding to the first stage,
(P4) is independent of the first three properties and holds with 50\%
probability.  Conditioned on (P1), (P2), and (P4) (as (P3) is
irrelevant), the probability of (P5) equals the probability that a
uniform draw from $[\frac {\del}{2k},\tfrac{\del}{k}]$ (used to determine the $H$-value)
is at least the $q$-value of the type B bidder, which is
conditionally distributed uniformly on
$(\frac{\del} {2k},\tfrac{\del}{k}]$.  This happens with
probability $\tfrac{1}{2}$.
We conclude that all of (P1)--(P5) hold with a positive
probability, namely
\begin{equation*}
\gam(\del)  =
\frac{\del^2}{32e^3}.
\end{equation*}
\end{proof}

To work toward statement~(ii), we next prove 
that, for every auction
strategy, conditioned on $\E$, the strategy fails to allocate the
item to the optimal bidder --- the type-B bidder with its
maximum-possible valuation --- with constant probability.
It suffices to analyze the auction strategy that, conditioned on $\E$,
maximizes the probability (over the remaining randomness) of
allocating to the optimal bidder --- of guessing, from among the two
bidders $j,\ell$ that in $\vals^{(m+1)}$ have valuation at least
$(\dist^{\alpha})^{-1}(1-\tfrac{\del}{k})$, which one is type A and which
one is type B.
Since the two bidders were
symmetric ex ante, Bayes' rule implies that the probability of guessing
correctly (given $\E$) is
maximized by, for every $\vals^{(1)},\ldots,\vals^{(m+1)}$,
choosing the scenario that maximizes the likelihood of
the valuation profiles $\vals^{(1)},\ldots,\vals^{(m+1)}$ (given
$\E$).

\begin{lemma}
\label{lem:wrong-choice}
Every auction strategy, conditioned on~$\E$, allocates to
a non-optimal bidder with probability at least $\tfrac{1}{3}$.
\end{lemma}
\begin{proof}
The only valuations that affect the relative likelihoods of the two
scenarios are
$\val^{(m+1)}_j$ and $\val^{(m+1)}_{\ell}$.
We already know the optimal bidder is either $j$ or $\ell$.
Property (P3) of event~$\E$ implies that the $m$ sample
valuations from $j$ and $\ell$ are equally likely to be generated
under the two scenarios --- the distributions of type-A and
type-B bidders differ only for quantiles in $[0,\tfrac{\del}{k}]$.

Now, given $\val^{(m+1)}_j$ and $\val^{(m+1)}_{\ell}$, the posterior
probabilities of the two scenarios are {\em not} equal.
The reason is
that, conditioned on $\E$, the type-A bidder's valuation is
distributed according $(F^{\alpha})^{-1}(q)$ where $q$ is uniform in
$[\frac{\del}{2k},\tfrac{\del}{k}]$, while the type-B bidder's
valuation is distributed according to the smaller of two i.i.d.\ such
samples.\footnote{In more detail, consider a type-B bidder $i$ and
  condition on the event that its quantile $q_i = 1 - \disti(\vali)$
  is in   $[\frac{\del}{2k},\tfrac{\del}{k}]$ and that its
  valuation is its maximum possible, which is equivalent to the
  condition that its fictitious quantile $q'_i$ that generates its
  threshold $H_i$ lies in $[q_i,\tfrac{\del}{k}]$.  The joint
  distribution of $(q_i,q'_i)$ is the same as the process that
  generates two i.i.d.\ draws from
  $[\frac{\del}{2k},\tfrac{\del}{k}]$ and assigns $q_i$ and
  $q'_i$ to the smaller and larger one, respectively.  Note that the
  valuation of the bidder is, by definition,
  $(F^{\alpha})^{-1}(1-q'_i)$.}
Thus, assigning the item to the bidder of $j,\ell$ with the lower
valuation (in $\vals^{(m+1)}$) maximizes the probability of allocating
to the optimal (type-B) bidder.  The probability that this allocation
rule erroneously allocates the item to the type-A bidder is the
probability that a sample for a distribution (the type-A bidder) is
smaller than the minimum of two other samples from the same
distribution (the type-B bidder), which is precisely~$\tfrac{1}{3}$.
\end{proof}

The following lemma completes the proof of Theorem~\ref{thm:lower}.
\begin{lemma}
\label{lem:rev-loss}
The revenue of every auction strategy is at most
the following fraction of an optimal auction's revenue:
\begin{align*}
 & 1 - \frac {1 - \ln 2} {96 e^3} \frac {1} { \ln \max\{ e, k\}} \min\{1, \frac {k}{e}\} \del^2 && \text{if}~ \alph = 1\\
& 1 - (1 -\alph 2^{1 - \alph})\frac {1}{96 e^3} \del^{1 + \alph} &&  \text{if}~ 0 < \alph < 1 ~\text{and}~ \frac 1k \le \qalph = \alph^{1/(1 - \alph)} \\
& 1 - (1 - \alph 2^{1 - \alph})\frac {1}{96 e^3} \del^{1 + \alph}
\frac {1} { ( k \alph^{1/(1 - \alph)} )^{\alph} }
&&  \text{if}~ 0 <\alph < 1 ~\text{and}~ \frac {\del}{2k} < \qalph = \alph^{1/(1 - \alph)} < \frac {1}{k}\\
& 1 - (1 - \alph 2^{1 - \alph})\frac {2^{\alph}}{96 e^3} \del
&&  \text{if}~ 0 <\alph < 1 ~\text{and}~ \qalph = \alph^{1/(1 - \alph)} \le \frac {\del}{2k}\\
& 1 - \frac {1}{96 e^3} \del &&  \text{if}~ \alph = 0.
\end{align*}
\end{lemma}
\begin{proof}
Condition on the event $\E$.
By~\eqref{eq:truncated},
the virtual value $\vv_B$ of the type B bidder $i$ equals $H_i \ge
(\dist^{\alph})^{-1}(\tfrac {\del}{k})$; 
substituting $q = \tfrac {\del}{k}$ in~\eqref{eqn:F-inv} yields a
lower bound on $H_i$ which implies that
\begin{equation*}
\vv_B ~\ge~ \left\{
\begin{array}{ll}
\frac {1} {1 - \alph} \left[ \left( \frac {k} {\del} \right) ^{1 - \alph} - 1 \right]  & \text{if}~ \alph < 1 \\
\ln \frac {k}{\del}  & \text{if}~ \alph = 1.
\end{array}
\right.
\end{equation*}
For the type A bidder, by~\eqref{eq:truncated}, the virtual value
$\vv_A$ is at most $\alph \cdot (\dist^{\alph})^{-1}(\tfrac {\del}{2k}) -1$; 
using~\eqref{eqn:F-inv} again, this implies
\begin{equation*}
\vv_A ~\le~ \left\{
\begin{array}{ll}
\alph \left[ \frac{1}{1 - \alph} \left[  \left( \frac {2k} {\del}\right) ^{1 - \alph} - 1 \right]  \right] -1
~=~ \frac{\alph \cdot 2^{1 - \alph}}{1 - \alph}  \left( \frac {k} {\del}\right) ^{1 - \alph} - \frac{1}{1 - \alph} & \text{if}~ \alph < 1 \\
\ln \frac {2k}{\del} - 1 & \text{if}~ \alph = 1.
\end{array}
\right.
\end{equation*}

Thus, still conditioned on $\E$,
\begin{equation*}
\pot_B - \pot_A ~\ge~ \left\{
\begin{array}{ll}
  \frac{1}{1 - \alph} \left( \frac {k} {\del}\right) ^{1 - \alph}\left( 1 - \alph \cdot 2^{1 - \alph}\right)  & \text{if}~ \alph < 1 \\
  1 - \ln 2 & \text{if}~ \alph = 1.
\end{array}
\right.
\end{equation*}

We now bound the fractional loss of revenue.
By Lemma~\ref{lem:of-of-i},
$\E$ occurs with probability at least $\del^2/(32e^3)$.
By Lemma~\ref{lem:wrong-choice},
conditioned on $\E$, a type A rather than a type B bidder is wrongly allocated the item
with probability $\tfrac 13$.
Thus the expected loss of revenue is at least
\begin{equation*}
\frac 13 \frac {\del^2} {32 e^3} (\pot_B - \pot_A).
\end{equation*}
Recall from
Lemma~\ref{lem:rev-upp-bdd} that the optimal revenue 
is bounded above by $R^*
= k \times \min \{ q^A_{\alpha}, \tfrac{1}{k} \} \cdot v^*$, where 
$v^* = (\dist^{A,\alpha})^{-1}(\max \{ 1 - \qalph^A,  \tfrac{k-1}{k}
\} )$
is the value corresponding to
quantile $\min\{ \qalph^A, \tfrac{1}{k} \}$ in $\dist^{A,\alpha}$.
Recalling~\eqref{eq:Rstar},
we can lower bound the fractional loss of revenue at follows.

If $\alph \in [0,1)$ and $ q_{\alph}^A \ge \frac 1k$, then the fractional
loss of revenue is at least 
\begin{equation*}
 \frac {\del^2} {3 \cdot 32 e^3} \frac {\pot_B - \pot_A} {R^*} ~=~
    \frac {1}{3 \cdot 32 e^3} (1 - \alph 2^{1 - \alph}) \frac { \frac{1}{1 - \alph} \left[ \left( \frac {k} {\del}\right) ^{1 - \alph}  \right] \del^2 }
      {\frac {1} {1 - \alph} (k^{1 - \alph} - 1)}
      ~ \ge ~  \frac {1 - \alph 2^{1 - \alph} } {96 e^3} \del^{1 +\alph}.
\end{equation*}

If $\alph \in [0,1)$ and $q_{\alph}^A \le \frac 1k$, then the
  fractional loss of revenue is at least 
\begin{eqnarray*}
 \frac {1} {3 \cdot 32 e^3} (1 - \alph 2^{1 - \alph}) \frac {\frac{1}{1 - \alph} \left[  \left( \frac {k} {\del}\right) ^{1 - \alph} \right]\del^2}
         {\frac {1} {1 - \alph} k \qalph^A \left[ \left( \frac{1} {\qalph^A} \right) ^{1 - \alph} -1 \right] }
& \ge & \frac { (1 - \alph 2^{1 - \alph}) \left( \frac {k} {\del} \right)^{1 - \alph} \del^2 } { 96e^3 k (q_{\alph}^A)^{\alph } }\\
& = & \frac {(1 - \alph 2^{1 - \alph}) \del^{1 + \alph} } {96e^3 (\qalph^A k)^{\alph} }.
\end{eqnarray*}
If $\qalph^A = \alph^{1/(1 - \alph)}$, this becomes
$$\frac {(1 - \alph 2^{1 - \alph}) \del^{1 + \alph} } {96e^3 (\alph^{1/(1 - \alph)} k)^{\alph} }$$
and if $\qalph^A = \tfrac {\del}{2k}$ this simplifies to
$$\frac{2^{\alph}(1 - \alph 2^{1 - \alph}) \del}{96e^3}.$$

Finally, if $\alph =1$, as $\qalph^A = \tfrac 1e$, the fractional loss of revenue is at least
\begin{equation*}
\frac{1 - \ln 2} {96e^3  \min\{1, \frac ke \} \ln \max\{e, k\} } \del^2.
\end{equation*}
\end{proof}

\section{The Upper Bound}
\label{sec:upper}

Section~\ref{ss:ema} describes in detail the empirical Myerson
auction, the auction strategy for which the guarantee in
Theorem~\ref{thm:emp-myer-perf} holds.

\subsection{The Empirical Myerson Auction}\label{ss:ema}

In the {\em empirical Myerson auction}, we assume we are given $m$
independent samples from each distribution $\Fi$.
The gist is to
treat the resulting empirical distribution as
the actual distribution in a Myerson auction (Section~\ref{s:prelim}),
though some additional technical details are required.
In our variant, a number of the samples with the highest values are
discarded, and there is a further detail regarding how to handle any
high bids that occur in the auction (i.e.\ bids larger than the
largest non-discarded sample).

In detail, for each bidder~$i$, we use the samples from $\Fi$ to
construct an ``empirical revenue curve'' as follows (see also Figure~\ref{f:curve}):
\begin{enumerate}

\item 
Suppose that the $m$ independent samples drawn from $\Fi$ have values
$v_{i1} \ge v_{i2} \ge \ldots \ge v_{im}$.
Define the ``empirical quantile'' of $v_{ij}$ as $\tfrac{2j-1}{2m}$.

\item 
Discard the $\floor{\hgxi m} - 1$ largest
samples, for a suitable $\hgxi > 0$.\footnote{The
  reason for discarding the largest samples is that if they were
present there is a non-negligible probability that they would create a
poor approximation at the high value end of the distribution, which is
the end that matters the most.  See also~\cite{DRY10}.}
Let~$S$ denote the remaining samples.

\item For each remaining sample $v_{ij} \in S$, plot a point 
$(\tfrac{2j-1}{2m}, \tfrac{2j-1}{2m} v_{ij})$.

\item Add points at (0,0) and (1,0).

\item
While only needed for the analysis,
it will be helpful to define the ``empirical revenue curve'', $\bRi(\bq)$:
this is the curve comprising straight-line
segments joining the sequence of points specified in Steps 3 and 4 above.

\item Take the convex hull --- the least concave upper bound --- of
 this point set.  Denote the resulting ``ironed empirical revenue curve'' by
$\bCRi$.  This curve has constant slope between any two consecutive
empirical quantiles of points of~$S$.

\end{enumerate}
\begin{figure}
\begin{center}
\includegraphics[scale=.6]{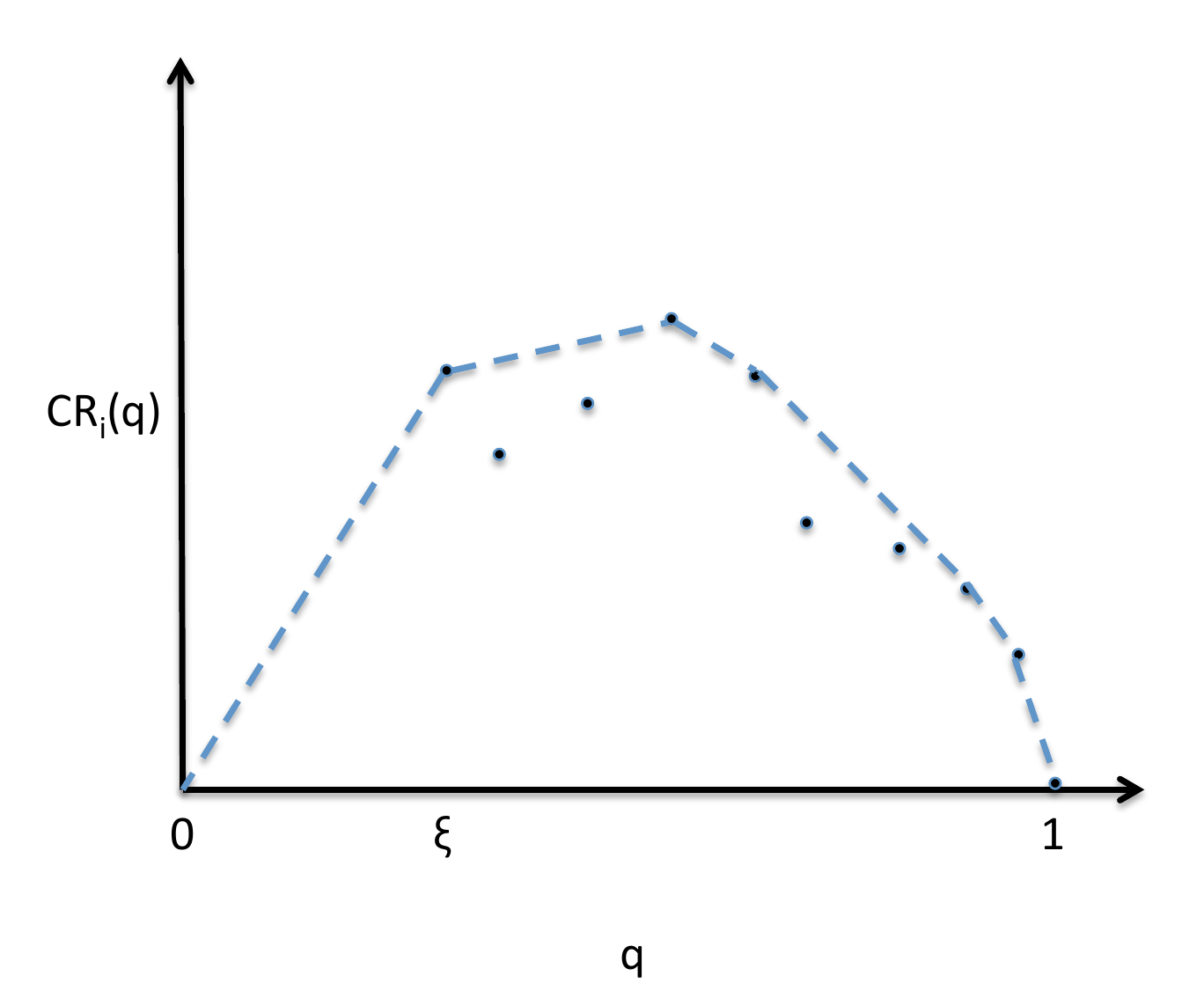}
\caption{Construction of the ironed empirical revenue curve.}
\label{f:curve}
\end{center}
\end{figure}

Now define empirical ironed virtual values as follows.  For
$v > v_{i,\hgxi m}$, it is simply the value~$v$.  For $v
\le v_{i,\hgxi m}$, identify the two samples $v_{ij},v_{i(j+1)} \in S$
that ``sandwich'' $v$.  The empirical ironed virtual value of $v$ is
defined as the slope of the revenue curve $\bCRi$ in the interval
defined by the empirical quantiles of $v_{ij}$ and
$v_{i(j+1)}$.\footnote{If $v$ is one of the points of~$S$ and there
  are multiple choices
  for this slope, we take the largest one.}
Note that the empirical ironed virtual value of $v_{i,\hgxi m}$
is also $v_{i,\hgxi
  m}$, and that the empirical ironed virtual valuation is a
nondecreasing function of~$v$.


Finally we run Myerson's auction on these empirical ironed virtual
valuations.
That is, the item is awarded to the bidder, if any, with the highest
non-negative virtual value (with ties broken arbitrarily).  The
winner's payment is the lowest bid needed to ensure a (tied) win. 

\subsection{Notation}

We next specify notation so as to clearly distinguish parameters for Myerson's optimal auction from
those for the empirical auction, as our analysis will be repeatedly comparing these two auctions.
After a couple of simple results,  Lemma~\ref{lem:samp-acc} bounds the empirical quantiles as a function of
the actual quantiles, and vice versa (this is essentially Lemma 4.1 in~\cite{DRY10}).
Next, Lemma~\ref{lem:phi-q-bdd}  relates the empirical and actual virtual values.
With these in hand, in Section~\ref{subsec:rev-loss-bound},
we bound the expected revenue loss due to using the empirical auction as opposed
to Myerson's optimal auction, assuming for the latter auction that the actual distributions were fully known.

\paragraph{Myerson's Auction}
Let $\MR$ (the ``Myerson Revenue'') denote the expected revenue 
earned by Myerson's auction.
Let $\exi(q)$ denote the probability that bidder $i$ wins in Myerson's auction with a bid
that has quantile $q$ in its value distribution.
Recall that $\vi(q)$ denotes the value corresponding to quantile $q$
and $\poti(q)$ denotes the virtual value at quantile $q$.
Let $\MRi = E_{\qi}[\poti \cdot \exi]$ denote the expected revenue
provided by $i$ in Myerson's auction (recall
Theorem~\ref{thm:Myerson}).
Let $\qi(v)$ denote the minimum quantile for value $v$; sometimes it will be convenient
to let $q^i_v$ denote $\qi(v)$.
Let $\ri$ be the reserve price applied to $i$ in Myerson's auction,
namely the largest value for which $\poti(\qi(v))=0$.
Let $\qri$ denote $\qi(\ri)$.
Let $\SRi = E[\poti(q) \, | \, q \ge \ri] =  \qri \cdot r_i$;
note that $\SRi$ is the expected revenue if $i$ were
the only participant in Myerson's auction ($\SRi$ is short for ``Single buyer Revenue'').
Sometimes, to reduce clutter,  we suppress the index $i$ when it is
clear from the context.  The following claim is immediate from the
definitions.

\begin{claim}
\label{clm:rev-facts}
\begin{enumerate}[i.]
\item
$\MR = \sum_{i=1}^n \MRi$.
\item
$\SRi \le \MR$ for all $i$.
\end{enumerate}
\end{claim}

\paragraph{The Empirical Auction}
The empirical auction is defined in terms of the ``empirical quantile''
$\bq$, but its analysis will entail considering its revenue as a
function of the actual quantile $q$. We specify notation which will distinguish between these two parameters.
For $v \le v_{i,\hgxi m}$, we define the empirical quantile $\bq(v)$
as the solution to $v \cdot q = \bRi(q)$ in~$q$.  (If there are multiple
solutions, we take the smallest one.)
Going the other way, for an empirical quantile $\bq \ge \bgxi$, we
define $\bv(\bq)$ as the solution to $v \cdot \bq = \bRi(\bq)$ in~$v$.
The empirical ironed virtual value $\bvi(\bq)$ of an empirical
quantile $\bq \ge \bgxi$ is the slope of $\bCRi$ at $\bq$.
%
Let $\bxi(\bq)$ denote the probability that bidder $i$ wins in the empirical
auction with the bid $\bvi(\bq)$.
Let $\bri$ denote the empirical reserve price, which is the minimum of
$\bvi(\bgxi)$ and the largest value $\bv$ for which
$\bpoti(\bqi(\bv))=0$, and let $\bqbri$ denote the corresponding empirical quantile.
Again, to reduce clutter, we sometimes suppress the index $i$ when it is clear from the context.

The actual quantile $q$ corresponding to empirical quantile $\bq$
is defined by the relation $\vi(q) = \bvi(\bq)$;
it is denoted by $q(\bvi(\bq))$; we write it as $q$ for short.
Finally, we write the empirical probability of winning as $\txi(q) = \bxi(\bq)$.

As the auction may draw values $\vi > \bvi(\bgxi)$ it is convenient for the purposes of our analysis
 to define $\bq$ for $\bq < \bgxi$.
Let $\xii$ be defined by $\vi(\xii) = \bvi(\bgxi)$.
For $\qi < \xii$, we define the corresponding value of
$\bq$, as $\bq(q) = \tfrac {\bgxi}{\xii} q$.
Then, for $ \bq < \bgxi$, $\bpoti(\bq) = \vi(q)$.

\hide{
\begin{lemma}
\label{lem:growth-of-v}
Let $F$ be a regular distribution.
Let $\qr \ge q_1 > q_2$, where $\qr$ is the quantile of the reserve price for $F$.
Then $v(q_2) \le \frac {q_1}{q_2} v(q_1)$.
\end{lemma}
\begin{proof}
$\pot(q) \ge 0$ for $q \le \qr$, and consequently $R(q_1) \ge R(q_2)$.
$R(q_1) = q_1 \cdot v(q_1)$ and $R(q_2) = q_2 \cdot v(q_2)$.
Thus $q_1 \cdot v(q_1) \ge q_2 \cdot v(q_2)$, and the result follows.
\end{proof}
}

\subsection{Relating the Actual and Empirical Quantiles}

The following result is essentially Lemma 4.1 in~\cite{DRY10}.
\begin{lemma}
\label{lem:samp-acc}
Let $F$ be a regular distribution.
Suppose $m$ independent samples with values $v_1 \ge v_2 \ge \ldots \ge v_m$ are drawn from $F$.
Let $ \gam >0$, $\hgxi = \frac am < 1$ for some integer $a > 0$ be given, and let
$\nu$ be defined by $1 + \nu = (1 + \gam)^2 $.
Let $\teeh = \frac {2h - 1} {2m}$.
Then, for all $v \le v_{\hgxi m }$,
\[
q(v) \in \left[ \frac {\bq(v)}  {(1 + \gam) ^2}, \bq(v)(1 + \gam)^2 \right]
= \left[ \frac {\bq(v)}  {(1 + \nu)}, \bq(v)(1 + \nu) \right]
\]
or equivalently
\[
\bq(v) \in  \left[ \frac {q(v)}  {(1 + \nu)}, q(v)(1 + \nu) \right]
\]
with probability at least $1 - \del$, if $\gam \hgxi m \ge 1$
and $m \ge \frac{6(1 + \gam)} {\gam^2  \hgxi} \max\{ \frac {\ln 3}{\gam},\ln \frac {3} {\del} \} $.
\end{lemma}
\begin{proof}
We begin by identifying a subsequence of the samples, $v_{l_1}, v_{l_2}, \ldots, v_{l_s}$,
with $l_1 \le l_2 \le \ldots \le l_s$; we rename the sequence
$u_1, u_2, \ldots, u_s$ for notational ease.
It will be the case that
$t_{l_{i+1}} \le (1 + \gam) t_{l_i}$, for $1 \le i < s$, and $t_{l_s} (1 + \gam) > 1$.
We will show that
\[
q(\vh) \in \left[ \frac {\teeh} {(1 + \gam) }, \teeh (1 + \gam) \right],~~~~\text{for~} \vh \in U = \{u_1, \ldots, u_s\}.
\]
The claimed bound is then immediate as either each $v  \le v_{\hgxi m }$ is sandwiched between two items in $U$, or it is at most $u_s$.

We define the $\li$ as follows: $l_1 = \hgxi m$ and $l_{i+1} = \floor{(1 + \gam)  \li}$
if $\floor{(1 + \gam)  \li} \le m$, and otherwise $l_{i+1}$ is not defined (i.e.\ $i = s$).
As $\gam \hgxi m \ge 1$,
$\floor{(1 + \gam) \hgxi m }\ge \hgxi m + \floor{\gam \hgxi m}
 \ge l_1 + 1$,
from which we conclude that the sequence is strictly increasing and hence well defined.

Next we bound the probability that $ q(\ui) > (1 + \gam) \tli$.
Now $q(\ui) > (1 + \gam) \tli$ only if fewer than $\li =  \tli m + \tfrac 12$ samples have $q$ values that are at most $(1 + \gam) \tli$.
As the expected number of such samples is $(1 + \gam) \tli m$,
a Chernoff bound gives the following upper bound on the probability that $q(\ui) > (1 + \gam) \tli$
(cf.~\cite{book-Mitzenmacher-Upfal}):
\begin{eqnarray*}
&& \exp\{- \frac {\gam^2 \tli m} {2 (1 + \gam)} \}.
\end{eqnarray*}
Similarly, the probability that $q(\ui) < \tli /(1 + \gam)$ is bounded by
\[
\exp\{- \frac {\gam^2 \tli m} {3 (1 + \gam)} \}.
\]
It will be helpful to bound both $ \tlone m$ and $[\tlip - \tli] m$.
As $ \hgxi m \ge 1$, $ \tlone m = (2 \hgxi m -1)/2 \ge \frac12 \hgxi m$.
And as $\gam \hgxi m \ge 1$, $[\tlip - \tli] m \ge \floor {(1 + \gam) \li} - \li \ge  \floor{\gam \li }
\ge  \floor{ \gam \hgxi m} \ge  \frac 12 \gam \hgxi m$.

Now, by the union bound applied to all the $q(\ui)$, we obtain a failure probability of at most:
\begin{align*}
  \sum_{i=1}^l & \exp\{ -\frac {\gam^2  \tli m} {3 (1 + \gam)} \}  + \exp\{ - \frac {\gam^2 \tli m} {2 (1 + \gam)} \} \\
 & \le ~ 2 \sum_{i=1}^l \exp\{- \frac {\gam^2 \tli m} {3 (1 + \gam)} \} \\
 & \le ~ 2 \sum_{i=0}^{l-1} \exp\left\{- \frac {\gam^2 [\hgxi m +(i-1)\gam \hgxi m] } {6 (1 + \gam)} \right\} \\
& ~~~~~~~~\text{using the bounds on } \tlone m  \text { and } [\tlip - \tli]m
\end{align*}
\begin{align*}
 & \le ~ \frac { 2 \exp\{- \frac {\gam^2 \hgxi m } {6 (1 + \gam)}\} }
                      { 1 - \exp\{ - \frac{ \gam^3 \hgxi m} {6 (1 + \gam)} \} }
 ~ \le ~ 3 \exp \left\{ - \frac {\gam^2 \hgxi m} {6 (1 + \gam)} \right\} \\
 &~~~~~~~~ \text{if } \exp\left\{- \frac {\gam^3 \hgxi m} {6 (1 + \gam) } \right\} \le \frac 13.
\end{align*}
We want the failure probability to be at most $\del$.
So we need $ \frac {\gam^2 \hgxi m}{6 (1 + \gam)} \ge \ln \frac {3}{\del}$,
i.e.\ $m \ge \frac {6 (1 + \gam) }{\gam^2 \hgxi} \ln \frac {3}{\del}$.
We also need $m \ge  \frac {6 (1 + \gam)}{\gam^3 \hgxi} \ln 3$ to satisfy the condition in the final inequality.
\end{proof}

\subsection{Relating the Actual and the Empirical Virtual Values}

Let ${\cal E}_a$ be the event that the high probability outcome of Lemma~\ref{lem:samp-acc}
occurs,
namely that for all $v \le v_{\floor{\hgxi m} }$,
$q(v) \in \left[ \frac {\bq(v)}  {(1 + \nu)}, \bq(v)(1 + \nu) \right]$.
${\cal E}_a$ occurs with probability at least $1 - \del$.
It will also be helpful to express the bound on $v$ as a bound on $\bq$.
To this end, we define $\bgxi = t_1 = \frac {2\hgxi m -1}{2m}$.

We will repeatedly encounter terms of the form $\pot(\lambda q)$ with $\lambda > 1$;
For $\lambda q >1$, $\pot(\lambda q)$ is interpreted to mean $\pot(1)$;
similarly for $\bpot$.

\hide{
\begin{lemma}
\label{lem:bpot-tpot-relation}
Conditioned on ${\cal E}_a$,
for all $v \le v_{\hgxi m }$,
\[
\tpot(\bq(v)(1+ \nu)) \le \bpot(\bq(v)) \le \tpot(\frac{\bq(v)} {1 + \nu}).
\]
\end{lemma}
\begin{proof}
By definition, $\tpot(q) = \bpot(\bq)$ for corresponding actual and empirical quantiles $q$ and $\bq$.
The result is now an immediate consequence of Lemma~\ref{lem:samp-acc}.
\end{proof}
}

\begin{lemma}
\label{lem:bCR-bound}
Conditioned on ${\cal E}_a$,
for all empirical quantiles $\bq \ge \bgxi$,
$\bCR(\bq) \le \bq \cdot v(\frac {\bq} {1 + \nu})$,
and for all $\teeh = \frac{2j-1} {2m} \ge \bgxi $,
$ \bCR(\teeh) \ge \teeh \cdot v(\ {\teeh} (1 + \nu))$.
\end{lemma}
This lemma is not as obvious as it may seem for it concerns points on the convex hull $\bCR$
of the set of points $\bR$ that are used to specify the empirical revenue.
\begin{proof}
By Lemma~\ref{lem:samp-acc}, as ${\cal E}_a$ holds,
for all $\teeh \ge \bgxi $,
\[
 \teeh \cdot v( \teeh(1 + \nu)) \le
\bR(\teeh) \le \teeh \cdot v(\frac {\teeh} {1 + \nu}).
\]

We define $\bL(\bq) = \bq \cdot v(\bq(1 + \nu))$ and $\bU(\bq) = \bq \cdot v( \frac{\bq} {1 + \nu})$ for
all $\bq$.

Note that for any pair $q \ne q'$ of quantiles,
the line joining the actual revenue
$R(\frac {q} {1 + \nu}) = \frac {q} {1 + \nu} v(\frac {q} {1 + \nu})$ to
$R(\frac {q'} {1 + \nu}) = \frac {q'} {1 + \nu} v(\frac {q'} {1 + \nu})$
is parallel to the line joining $\bU(q)$ to $\bU(q')$,
for the latter line is obtained by expanding the former line by a factor $1 + \nu$
in both the quantile and revenue dimensions.
By the regularity of $\pot$,
the curve defined by $R$ is convex, and
consequently,
the points $\bU(\bq)$ all lie on their convex hull.

For $\teeh  \ge \bgxi$,
$\bU(\teeh)$ is an upper bound on $\bR(\teeh)$;
it follows that the convex
hull for the empirical revenue, for $\bq\ge \bgxi$,
is enclosed by the convex hull $\bU(\bq)$,
and consequently $\bCR(\bq) \le \bU(\bq) = \bq \cdot v(\frac {\bq} {1 + \nu})$.

For the second result, the lower bound, we use a similar argument, but it will
apply just to the empirical quantiles $\teeh \ge \bgxi$.
Now, for any pair $q \ne q'$ of quantiles,
the line joining the actual revenue
$R(q (1 + \nu)) = q (1 + \nu) v(q (1 + \nu))$ to
$R(q' (1 + \nu)) = q' (1 + \nu) \cdot v(q' (1 + \nu))$
is parallel to the line joining $\bL(q)$ to $\bL(q')$, and hence the points
$\bL(\bq)$ all lie on their convex hull.
But, for $\teeh \ge  \bgxi$, $\bL(\teeh) \le \bR(\teeh)$, and
consequently the values $\bR(\teeh)$ all lie on or above the curve $\bL(\teeh)$.
\end{proof}

The following lemma, which lies at the heart of out analysis, shows that
with high probability$\pot(q)$ is close to some value $\bpot(\bqp)$ with
$\bqp \in [\frac{\bq} {(1 + \Del) (1 + \nu)^3}, \bq (1 + \Del)(1 + \nu)]$.

\begin{lemma}
\label{lem:phi-q-bdd}
Let $F$ be a regular distribution.
Suppose that $(1 + \Del) \ge (1 + \nu)^2$.
Let $\teeh = \frac{2h-1}{2m}$, for $1 \le h \le m$.
Conditioned on ${\cal E}_a$,
if  $t_{h-1} < \bq \le t_h$, then
\begin{enumerate}[i.]
\item
for all $\bq$ with $\bgxi (1 + \Del)(1 + \nu)^3 \le \bq$,
$ \pot(q) \le  \pot (\frac {\teeh} {(1 + \nu)^2}) \le \bpot ( \frac{\bq} {(1 +\Del)(1 + \nu)^3} )+ 2 \frac {\nu}{\Del} (1 + \Del)(1 + \nu)^3 \frac {\SR}{\bq} $,
and
\item
for all $\bq$ with  $\bgxi \le \bq$,
$\bpot(\bq(1 + \Del)(1 + \nu)) \le  \pot(\teeh (1 + \nu)) + 2 \frac {\nu}{\Del} \frac {\SR}{\bq}
\le  \pot(q) +  2 \frac {\nu}{\Del}\frac {\SR}{\bq} $.
\end{enumerate}
\end{lemma}
\begin{proof}
The main part of the proof concerns the second inequality in (i) and the first one in (ii).
We begin by proving the inequality in (i).
First we give an upper bound on $\pot( \frac{\teeh} {(1 + \nu)^2} )$
and a lower bound on $\bpot(\frac{\bq }{(1 + \Del) (1 + \nu)^3})$.

As $F$ is regular, $R$ is convex; thus:
\begin{eqnarray*}
\pot( \frac {\teeh} {(1 + \nu)^2} ) &\le & \frac{ R( \frac {\teeh} {(1 + \nu)^2} ) - R(\frac {\teeh } {(1 + \Del) (1 + \nu)^4 }) }
                       { \frac {\teeh} {(1 + \nu)^2} - \frac {\teeh } {(1 + \Del) (1 + \nu)^4 } }\\
&=& \frac{ \frac {\teeh} {(1 + \nu)^2} \cdot v(\frac {\teeh} {(1 + \nu)^2})
          - \frac {\teeh } {(1 + \Del) (1 + \nu)^4 } v(\frac {\teeh } {(1 + \Del)  (1 + \nu)^4 ) })
      }
      { \frac {\teeh} {(1 + \nu)^2} - \frac {\teeh } {(1 + \Del) (1 + \nu)^4 } }\\
& = &
\frac{ (1 + \Del)(1 + \nu)^2  v( \frac {\teeh} {(1 + \nu)^2}) -  v(\frac {\teeh } {(1 + \Del) (1 + \nu)^4 } ) }
     {  2\nu + \nu^2 + \Del (1 + \nu)^2}.
\end{eqnarray*}

The following bound applies only when $\frac{\teeh} {(1 + \Del)  (1 + \nu)^3} \ge \bgxi$ for otherwise
$\bCR( \frac {\teeh} {(1 + \Del) (1 + \nu)^3} )$ is not defined;
the constraint $\bq \ge \bgxi (1 + \Del) (1 + \nu)^3 $ suffices.
\begin{eqnarray*}
\bpot( \frac {\bq } {(1 + \Del)  (1 + \nu)^3  })
& \ge &
 \frac{ \bCR( \frac {\teeh}  {(1 + \nu)^3} ) - \bCR(\frac {\teeh} {(1 + \Del) (1 + \nu)^3}) }
                      { \frac {\teeh}  {(1 + \nu)^3} - \frac {\teeh} {(1 + \Del) (1 + \nu)^3} } \\
& \ge &
   \frac{ \frac {\teeh}  {(1 + \nu)^3} \cdot v( \frac {\teeh}{ (1 + \nu)^2 } )
              -  \frac {\teeh} {(1 + \Del) (1 + \nu)^3} v( \frac {\teeh } { (1 + \Del) (1 + \nu)^4 } )}
         { \frac {\teeh}{ (1 + \nu)^3 }  - \frac {\teeh} { (1 + \Del) (1 + \nu)^3 } }~~~~\text{(by Lemma~\ref{lem:bCR-bound})}\\
& = &
\frac{ (1 + \Del)  v( \frac {\teeh} { (1 + \nu)^2 } ) - v( \frac {\teeh } { (1 + \Del) (1 + \nu)^4 } ) }
     { \Del}.
\end{eqnarray*}

Now, we combine the bounds so as to eliminate the term $v( \frac {\teeh } { (1 + \Del) (1 + \nu)^4 } )$.

\begin{align*}
& \frac {  2\nu + \nu^2 + \Del (1 + \nu)^2} {\Del} \pot( \frac {\teeh} { (1 + \nu)^2 } ) - \bpot(\frac{\bq } { (1 + \Del) (1 + \nu)^3 }) \\
& ~~~~\le~ \frac{ (1 + \Del)[(1 + \nu)^2 -  (1 + \Del)] v(  \frac {\teeh} { (1 + \nu)^2})   } {\Del }\\
& ~~~~\le~ \frac{ (1 + \Del) (2\nu + \nu^2) } {\Del } \frac {\SR  (1 + \nu)^2}{\teeh}  ~~~(\text{as}~\SR \ge \frac {\teeh} {(1 + \nu)^2} v \left( \frac {\teeh} {(1 + \nu)^2} \right)~) \\
& ~~~~\le~ 2 \frac{\nu} {\Del} (1 + \Del) (1 + \nu)^3 \frac {\SR} {\bq}.
\end{align*}

In other words,
\begin{equation*}
\left( 1 + \frac {\nu (1 + \Del) (2 + \nu)} {\Del} \right) \pot( \frac {\teeh} { (1 + \nu)^2 } ) - \bpot(\frac{\bq } { (1 + \Del) (1 + \nu)^3 }) ~\le~ 2 \frac{\nu} {\Del} (1 + \Del) (1 + \nu)^3 \frac {\SR} {\bq}.
\end{equation*}
Thus
\begin{equation*}
\pot( \frac {\teeh} { (1 + \nu)^2 } ) \le \bpot(\frac{\bq } { (1 + \Del) (1 + \nu)^3 }) + 2 \frac{\nu} {\Del} (1 + \Del) (1 + \nu)^3 \frac {\SR} {\bq}.
\end{equation*}

The second inequality in (ii) is shown similarly.
We start with an upper bound on $\bpot(\bq (1 + \Del) (1 + \nu)) $ and a lower bound on $\pot( \teeh (1 + \nu))$.
The first bound applies only when $\teeh \ge \bgxi$; here $\bq \ge \bgxi$ suffices.

\begin{eqnarray*}
\bpot(\bq (1 + \Del) (1 + \nu)) ~ \le ~ \bpot(\teeh (1 + \Del) ) & \le & \frac{ \bCR(\teeh (1 + \Del) ) - \bCR(\teeh) }
                       { (1 + \Del) \teeh - \teeh } \\
& \le & \frac{ \teeh (1 + \Del) \cdot v( \frac{\teeh (1 + \Del)} {1 + \nu} )
          - \teeh  v( \teeh  (1 + \nu)  )
      }
      {\Del \teeh  } ~~~~\text{(by Lemma~\ref{lem:bCR-bound})} \\
& = &  \frac{  (1 + \Del)  v(\frac {\teeh (1 + \Del) } {1 + \nu}) -  v( \teeh (1 + \nu)) }
        { \Del }.
\end{eqnarray*}

\begin{eqnarray*}
\pot( \teeh (1 + \nu) ) ~\ge ~ \frac{ R( \frac {\teeh (1 + \Del) } {(1 + \nu)}) - R( \teeh (1 + \nu) ) }
                       {  \frac {\teeh (1 + \Del) } {(1 + \nu)} - \teeh (1 + \nu) }
& =& \frac{ \frac {\teeh (1 + \Del) } {1 + \nu}  v(\frac {\teeh (1 + \Del) } {1 + \nu})
          - \teeh (1 + \nu) v( \teeh (1 + \nu) )
      }
      { \frac {\teeh (1 + \Del)} {1 + \nu} - \teeh (1 + \nu) }\\
& = &
\frac{ (1 + \Del)  v(\frac {\teeh (1 + \Del) } {1 + \nu}) - (1 + \nu)^2 v( \teeh (1 + \nu) ) }
     { \Del - 2\nu - \nu^2}.
\end{eqnarray*}

Again, we combine the bounds so as to eliminate the term $v( \teeh (1 + \nu) )$.

\begin{align*}
& \bpot(\bq (1 + \Del) (1 + \nu)) - \frac  {\Del - 2\nu - \nu^2} {\Del (1 + \nu)^2} \pot( \teeh (1 + \nu) )\\
& ~~~ \le ~ \frac{ (1 + \Del)[ (\Del - 2\nu - \nu^2) - \frac  {\Del - 2\nu - \nu^2}  {(1 + \nu)^2}]  v(\frac {\teeh (1 + \Del) } {1 + \nu}) }
            {\Del (\Del - 2\nu - \nu^2) }\\
& ~~~ \le ~  \frac{ (1 + \Del)  \nu( 2 + \nu)} {\Del (1 + \nu)^2}    \frac{SR (1 + \nu) } { (1 + \Del)\teeh} ~~~(\text{as}~\SR \ge \frac {(1 + \Del) \teeh} {1 + \nu} v \left( \frac {(1 + \Del) \teeh} {1 + \nu} \right)~) \\
&  ~~~ \le ~  2 \frac{\nu}{\Del} \frac{\SR} {\bq}.
\end{align*}

In other words,
\begin{equation*}
\bpot(\bq (1 + \Del) (1 + \nu)) - \left(1 - \frac{\nu (1 + \Del) (2 + \nu)} {\Del(1 + \nu)^2}\right) \pot( \teeh (1 + \nu))
~ \le ~ 2\frac{\nu}{\Del} \frac{\SR} {\bq}.
\end{equation*}
Thus
\begin{equation*}
\bpot(\bq (1 + \Del) (1 + \nu)) \le \pot( \teeh (1 + \nu)) + 2\frac{\nu}{\Del} \frac{\SR} {\bq}.
\end{equation*}

We now show the remaining inequalities.
To obtain the first inequality in (i), we note that
by Lemma~\ref{lem:samp-acc} and ${\cal E}_a$,
$q \ge \frac{\bq} {(1 + \nu)} > \frac {\teeh} { (1 + \nu)^2 }$,
from which the result follows.
Similarly, for the second inequality in (ii),
$\teeh (1 + \nu) \ge \bq (1 + \nu)\ge q$,
and again the result follows.
\end{proof}

\subsection{Bounding the Expected Revenue Loss}
\label{subsec:rev-loss-bound}

Finally, we consider an auction with $k$ bidders, where the valuation for the $i$th bidder
comes from regular distribution $\Fi$.
For brevity, bidder $i$ is referred to as $i$.

We define ${\cal E}_b$ to be the event that ${\cal E}_a$ holds for every distribution $\Fi$.

Let $\Shtfl = \sum_i E[\poti\cdot\exi] - \sum_i E[\poti \cdot \txi]$.
In other words, $\EMR + \Shtfl = \MR$, so it suffices to show that $\Shtfl \le \eps \MR$.
Recall that $\qri$ denotes the quantile of $\Fi$ corresponding to the reserve price for $i$ in the Myerson auction
and $\qbri$ denotes the quantile corresponding to the reserve price in the empirical auction.
Also, we let $\bqi$ be a quantile for $i$ in the empirical auction, and we  let $\qi$ denote the corresponding quantile in $\Fi$.
$\bqj$ and $\qj$ are defined similarly with respect to $j$.
In addition, to reduce clutter, we let $\beta = (1 + \Del)(1 + \nu)^3 -1$.

The next lemma provides an upper bound on $\Shtf$ as the sum of several terms
which we will bound in turn.
\begin{lemma}
\label{lem:shtfl-contr}
Conditioned on ${\cal E}_b$,
\begin{eqnarray}
\nonumber
\text{\rm Shtf} & = & 
\label{eqn:shtfl-term-one}  
\sum_i \left[ \int_{\qi \le \qri} \poti(\qi)\cdot \exi(\qi) \dqi
    -  \int_{\qi \le \qbri} \poti(\qi)\cdot \txi(\qi) \dqi \right]
\\
\label{eqn:shtfl-term-two}
& \le & \sum_i \int_{\qi \le \qri} \poti(\qi) \cdot [\exi(\qi) - \txi(\frac{\qi} {1 + \beta})]  \dqi \\
\label{eqn:shtfl-term-three}
 &&~~~~~~~~~~~~   + \beta \sum_i \int_{\qi \le \qri} \poti(\qi)  \dqi \\
\label{eqn:shtfl-term-four}
&&~~~~~~~~~~~~   + \sum_i \int_{ \qri \le \qi \le \qbri } [-\poti(\qi) ] \dqi.
\end{eqnarray}
\end{lemma}
\begin{proof} We upper bound the second (negative) term in the expression for $\Shtf$.
\begin{align*}
-\int_{\qi \le \qbri} \poti(\qi) \cdot \txi(\qi) \dqi
& = -\int_{\qi \le \qri} \poti(\qi) \cdot \txi(\qi) \dqi + \int_{ \qri \le \qi \le \qbri } [-\poti(\qi) ] \cdot \txi(\qi)\dqi\\
& \le -\int_{\qi \le \qri/(1 + \beta)} \poti(\qi)\cdot \txi(\qi) \dqi + \int_{ \qri \le \qi \le \qbri } [-\poti(\qi) ] \dqi
\end{align*}
and
\begin{align*}
& -\int_{\qi \le \qri/(1 + \beta)} \poti(\qi) \cdot \txi(\qi) \dqi \\
&\hspace*{1in} =- (1 + \beta)\int_{\qi \le \qri/(1 + \beta)} \poti(\qi) \cdot \txi(\qi) \dqi + \beta \int_{\qi \le \qri/(1 + \beta)} \poti(\qi) \cdot\txi(\qi) \dqi \\
&\hspace*{1in}  \le  -\int_{\qi \le \qri} \poti(\frac {\qi}{1 + \beta})\cdot \txi(\frac {\qi}{1 + \beta}) \dqi + \beta \int_{\qi \le \qri} \poti(\qi) \cdot \txi(\qi)\dqi \\
&\hspace*{1in} \le - \int_{\qi \le \qri} \poti(\qi) \cdot \txi(\frac {\qi}{1 + \beta}) \dqi
+ \beta\int_{\qi \le \qri} \poti(\qi)  \dqi.
\end{align*}
Substituting in the expression for $\Shtf$ yields the result.
\end{proof}

The bound on ~\eqref{eqn:shtfl-term-three} is simply
\begin{equation}
\label{eqn:shtfl-frst-bnd}
\beta \sum_i \SRi = k \beta \cdot \MR.
\end{equation}
In the following lemmas we bound the terms
\eqref{eqn:shtfl-term-four}
and \eqref{eqn:shtfl-term-two}.
To bound \eqref{eqn:shtfl-term-four} we partition the integral into two intervals.
The intervals are the ranges $\qri \le \qi \le \max\{\xii, \qri\}$ and $\max\{\xii, \qri\} \le \qi \le \qbri$, respectively,
where $\xii$ is the quantile of $\Fi$ corresponding to empirical quantile $\bgxi$.

\begin{lemma}
\label{lem:frth-term-i}
Conditioned on ${\cal E}_b$,
\[
\sum_i \int_{ \qri \le \qi \le \max\{\xii, \qri\}}  [-\poti(\qi)] \dqi \le
k \bgxi (1 + \nu)  \cdot \MR.
\]
\end{lemma}
\begin{proof}
If $\xii \le \qri$ the integral is zero and the result is immediate. So we can assume that $\xii \ge \qri$.
Note that $-\poti(\qi)$ is a non-decreasing function of $\qi$; thus its smallest values in the
range $\qi \ge \qri$ occur in the integral we are seeking to bound. It follows that
\begin{align*}
\int_{ \qri \le \qi \le \xii }  [-\poti(\qi)] \dqi
& \le  \frac{\xii - \qri} {1 - \qri} \int_{ \qri \le \qi}  [-\poti(\qi)] \dqi \\
& \le \xii \int_{ \qri \le \qi}  [-\poti(\qi)] \dqi \\
& \le  \xii \int_{\qi \le \qri} \poti(\qi) \dqi ~~~~~~~~~~~~~~(\text{as}~\int_{0 \le \qi \le 1} \poti(\qi) \dqi = 0) \\
& =  \xii \cdot \SRi \le \xii \cdot \MR \le \bgxi(1 + \nu)\cdot \MR.
\end{align*}
The last two inequalities follow from Claim~\ref{clm:rev-facts}(ii) and Lemma~\ref{lem:samp-acc},
respectively.
The result follows on summing over $i$.
\end{proof}

\begin{lemma}
\label{lem:fifth-term}
Conditioned on ${\cal E}_b$,
\[
E\left[ \sum_i \int_{\max\{\xii, \qri\}\le \qi \le \qbri} [-\poti(\qi) ] \dqi \right] \le
2 \nu \sum_i \SR_i \le 2k\nu \cdot \MR.
\]
\end{lemma}
\begin{proof}
let $\chii = \max\{\xii, \qri\}$ and let $\bchii = \max\{\bgxi, \bqri\}$ be the corresponding empirical quantile.
Again, if $\chii \ge \qbri$ the integral is zero and the result is immediate.
So we can assume that $\chii < \qbri$.
The derivation below uses
Lemma~\ref{lem:samp-acc} to justify the first and third inequalities;
the second inequality follows from the definition of $\bri$ as the empirical reserve price
since $\bchii \ge \bgxi$.
Conditioned on ${\cal E}_b$,
\begin{equation}
\label{eq:emp-rev-bdd}
\qbri \cdot  \bri ~\ge~
\frac {\bqbri \cdot  \bri} {1 + \nu}
~\ge~  \frac{ \bchii \cdot \bvi(\bchii)} {1 + \nu}
~\ge~ \frac {\chii \cdot \vi(\chii)} {(1 + \nu)^2}.
\end{equation}

Thus
\begin{eqnarray*}
\int_{\max\{\xii, \qri\} \le \qi \le \qbri}  [-\poti(\qi)] \dqi
 & = & \chii \cdot \vi(\chii) - \qbri \cdot \bri
~ \le ~  \chii \cdot \vi(\chii) \left[ 1 - \frac {1} {(1 + \nu)^2} \right] ~~~\text{(by \eqref{eq:emp-rev-bdd})}\\
& \le & \frac {\nu (2+ \nu) } {(1 + \nu)^2} \SRi  \le 2 \nu \cdot \SRi \le 2 \nu \cdot \MR.
\end{eqnarray*}
\end{proof}

It remains to bound term \eqref{eqn:shtfl-term-two}.

Before proceeding to the next lemma we need some additional terminology, namely the notions of
$i$-safety, and of large and small amounts.

\begin{definition}
\label{def:i-safe}
The vector of quantiles $q=(q_1, q_2, \cdots, q_k)$ is said to be $i$-\emph{safe} if
$\bqh \ge \bgxi$ for all $h \ne i$.
\end{definition}
We will also write $\q = (\qi, \qj, \nqij)$
and $\q = (\qi, \nqi)$, when we want to focus on just two or one coordinates of the quantile vector.

We define large and small amounts with respect to $\poti(\qi)$ and $\potj(\qj)$ as follows.
\begin{definition}
\label{def:large-small}
Let $\rho$ and $\rho'$ be defined as in Lemma~\ref{lem:samp-acc}.
Suppose that $\poti(\qi) \ge \potj(\qj)$.
$\poti(\qi)$ is said to exceed $\poti(\qj)$ by a \emph{large amount} in the following cases:
\begin{enumerate}[i.]
\item
\label{list:large-mnt-main}
For $\bqi \ge \bgxi(1+ \beta)$,
\[
\poti(\qi) - \poti(\qj) \ge
2 \frac {\nu}{\Del} (1 + \beta) \cdot \frac {\SRi}{\bqi}  + 2 \frac {\nu}{\Del} \frac {\SRj}{\bqj},
\]
\item
\label{list:large-mnt-scnd}
and for $\bqi < \bgxi(1+ \beta)$,
\[
\poti(\qi) - \poti(\qj) \ge 2 \frac {\nu}{\Del} \frac {\SRj}{\bqj}.
\]
Otherwise, $\poti(\qi)$ is said to exceed $\poti(\qj)$ by a \emph{small amount}.
\end{enumerate}
\end{definition}

The following lemma bounds the probability that $i$ wins by a large amount over $j$ in the Myerson
auction at quantile $\q$, while $j$ wins in the empirical auction at quantile $\q/(1 + \beta)$.
\begin{lemma}
\label{lem:cond-loss}
Conditioned on ${\cal E}_b$,
for any pair $i$ and $j$, and for any $\qi$,
if $\bqj \ge \bgxi$, then
the probability of the following event is bounded by
$(1 + \beta)^2 - 1$:
\begin{quote}
$i$ wins in the Myerson auction by a large amount over $j$ at quantile $\q$,
and $j$ wins in the empirical auction at quantile $\bbq/(1 + \beta)$,
where $\q = (\qi, \nqi)$ is $i$-safe.
\end{quote}
\hide{
\begin{align*}
& (1 + \rho') \potj(\qj) + 2 \frac {\nu}{\Del}  \frac {\SRj} {\bqj}
< (1 - \rho) \poti(\qi) - 2 \frac {\nu}{\Del} (1 + \beta) \cdot \frac {\SRi}{\bqi}  ~\text{and} \\
& \bpoti(\frac{\bqi}{1 + \beta})  < \bpotj(\frac{\bqj}{1 + \beta}).
\end{align*}

\smallskip
\noindent
ii. While if $\bqi < \bgxi (1 + \beta)$, the event is
\begin{align*}
 & (1 + \rho') \potj(\qj) + 2 \frac {\nu} {\Del}  \frac {\SRj} {\bqj} < \poti(\qi), ~\text{and} \\
 & \bpoti(\frac{\bqi} {1 +\beta})  < \bpotj(\frac{\bqj}{1 + \beta}).
\end{align*}
}
\end{lemma}
\begin{proof}
We begin with the case that $\bqi \ge \bgxi (1 + \Del)$.
Given ${\cal E}_b$, by Lemma~\ref{lem:phi-q-bdd},
for $ \bgxi (1 + \beta) \le \bqi$,
$\poti(\qi) \le \bpoti( \frac{\bq}{1 +\beta}) +  2 \frac {\nu}{\Del} (1 + \beta) \cdot \frac {\SRi}{\bqi}$
and for $\bgxi \le \bqj  $,
$\bpotj(\bqj (1 + \Del) (1 + \nu))  \le \potj({\qj} ) + 2 \frac {\nu}{\Del} \frac {\SRj}{\bqj} $.
Thus, 
\begin{align*}
 \bpotj(\bqj (1 + \Del)(1 + \nu))
& \le  \potj(\qj) + 2 \frac {\nu}{\Del} \frac {\SRj}{\bqj} \\
& < \poti(\qi) - 2 \frac {\nu}{\Del} (1 + \beta) \cdot \frac {\SRi}{\bqi} 
~~~ \mbox{(from Definition~\ref{def:large-small}(\ref{list:large-mnt-main}))}\\
& \le \bpoti( \frac{\bqi}{1 + \beta}) \\
& < \bpotj( \frac{\bqj}{1 + \beta}) .
\end{align*}
Thus we have a lower bound of $\bpotj(\bqj (1 + \Del)(1 + \nu))$ and an upper bound
of $\bpotj( \frac{\bqj}{1 + \beta})$ on the remaining terms.
Clearly these can both hold only for a limited range of $\bqj$ and hence of $\qj$, which we bound as follows.
Define $\hbqj = \arg\inf_{\bqj}\{ \bpotj(\bqj (1 + \Del)(1 + \nu)) \le \bpoti(\bqi/[1 + \beta])\}$.
Then these bounds can hold at most for $\bqj$ satisfying
$\hbqj
\le \bqj <
\hbqj(1 + \Del) (1 + \nu) (1 + \beta)
$.
To obtain a probability bound, one needs to express the range in  terms of the $\qj$ quantile,
namely ranging at most from $\hbqj /(1 + \nu) $ to $\min\{1, \hbqj (1+ \Del) (1+\nu)^2 (1 + \beta) \}$,
i.e.\ with probability at most $(1 + \Del)(1+\nu)^3(1 + \beta)  -1 =   (1 + \beta)^2 - 1$.

When $\bqi < \bgxi(1 + \beta)$, we proceed similarly.
(The third inequality below follows because for $\bqi \le \bgxi$,
$\bpoti(\bqi) = \vi(\qi)$, and the fourth inequality holds because
$\frac {\bqi} {(1 + \beta)} \le \bgxi$ by assumption.)
\begin{align*}
 \bpotj(\bqj(1 + \Del)(1 + \nu))
& \le \potj({\qj} ) + 2 \frac {\nu}{\Del} \frac {\SRj} {\bqj}  < \poti(\qi) \\
& \le \max\left\{\vi (\qi), \vi(\xii) \right\} \\
& \le \bpoti(\min\{ \bqi, \bgxi \}) \\
& \le  \bpoti(\frac {\bqi} {1 + \beta}) < \bpotj(\frac {\bqj} {1 + \beta}).
\end{align*}
The rest of the argument is as for (i).
\end{proof}

\begin{lemma}
\label{lem:scnd-term}
Conditioned on ${\cal E}_b$,
\begin{align*}
& \sum_i \int_{\qi \le \qri} \poti(\qi) \cdot [\exi(\qi) - \txi(\frac{\qi} {1 + \beta})]  \dqi\\
&~~~~\le \left[  (k-1)\bgxi (1 + \nu)  +  k (k-1) [(1 + \beta)^2 -  1]  + (\rho + \rho')
+4 k (1 + \beta) (1 + \nu)\frac {\nu}{\Del} \cdot \ln \frac{1 + \nu} {\bgxi} \right]  \MR.
\end{align*}
\end{lemma}
\begin{proof}
We let $\Evix(\qi)$ be the event that quantile $\q = (\qi, \nqi)$ is not $i$-safe.
Clearly $\Pr[ \Evix]  \le \sum_{j \ne i} \xij \le (k-1) \bgxi (1 + \nu)$.
Let $\xie(\qi)$ be the probability that $i$ wins when the quantile $\q= (\qi, \nqi)$ is not $i$-safe.
Note that $\exi(\qi\,|\,\mbox{$\q$ is not $i$-safe}) \le \exi(\qi)$, as having some $\qj$ be small
only increases the probability that $j$ wins.
Thus
\begin{equation}
\label{eqn:extreme-prob-bnd}
\xie(\qi) \le (k-1) \bgxi (1 + \nu) \exi(\qi).
\end{equation}

We let $\xib$ denote the probability that $i$ wins both
in Myerson's auction with $i$-safe $\q =(\qi, \nqi)$
and in the empirical auction at quantile $\q/(1 + \beta)$.

We also introduce notation to measure the probability of wins by small and large amounts, for
$i$-safe quantiles.
We will be measuring the probability that $i$ wins in the Myerson auction at quantile
$\q = (\qi, \qj, \nqij)$ and $j$ wins in the empirical auction at quantile $\q/(1 + \beta)$,
for some $\nqij$.
$\xijs(\qi,\qj)$ measures the probability of this event in the case that the win
in the Myerson auction is by a small amount,
and $\xijl(\qi, \qj)$ measures the probability of the event when the win margin is large.

Switching perspectives, we let $\txijs(\qi,\qj)$ denote the probability that
$i$ wins in the empirical auction at quantile $\q/(1 + \beta)$ and $j$
wins by a small amount over $i$ in the Myerson auction at quantile $\q$, where
$\q = (\qi, \qj, \nqij)$ is $j$-safe. Clearly,
$\txjis(\qj, \qi) = \xijs(\qi,\qj)$.

We also note that
\begin{equation}
\label{eq:i-sum-phi}
\exi(\qi) = \xie(\qi) + \xib(\qi)  
 + \sum_{j\ne i} \int_{\bqj \ge \bgxi} \xijs(\qi,\qj) + \xijl(\qi,\qj)  \dqj,
\end{equation}
and for $\bqi \ge \bgxi$,
\begin{equation}
\label{eq:j-sum-phi}
\txi(\frac{\qi} {1 + \beta}) \ge \xib(\qi) + \sum_{i \ne j} \int_{\bqj \ge 0} \txijs(\qi,\qj) \dqj.
\end{equation}

By Lemma~\ref{lem:cond-loss},
\begin{equation}
\label{prob-large-win-bnd}
\int_{\qj \ge \bgxi} \xijl(\qi,\qj) \le (1 + \beta)^2 -1.
\end{equation}

We obtain:

\begin{align}
\nonumber
& \sum_i \int_{\qi \le \qri} \poti(\qi) \cdot [\exi(\qi) - \txi(\frac{\qi} {1 + \beta})]  \dqi \\
\nonumber
& ~~~ \le \sum_i \int_{\qi \le \qri} \poti(\qi) \left[ \left(\xib(\qi) + \xie(\qi) 
 + \sum_{j\ne i} \int_{\bqj \ge \bgxi} \xijs(\qi,\qj) + \xijl(\qi,\qj) \right) \dqj  \right.\\
 \nonumber
& \hspace*{1.35in} - \left. \left( \xib(\qi) + \sum_{j\ne i} \int_{\bqj \ge \bgxi} \txijs(\qi, \qj) \right) \dqi \right] \dqi~~~~\mbox{(using~\eqref{eq:i-sum-phi} and~\eqref{eq:j-sum-phi})}
 \\
\nonumber
& ~~~ \le  \sum_i \int_{\qi \le \qri} \poti(\qi)
\left[ \xie(\qi) + \sum_{j \ne i}\int_{\bqj \ge \bgxi} [\xil(\qi, \qj) + \xis(\qi,\qj)]\right] \dqi \dqj \\
& ~~~~~~~~ - \sum_j \int_{\qj \le \qrj} \potj(\qj) 
         \cdot \sum_{i \ne j} \int_{\bqi \ge \bgxi} \txjis(\qj,\qi) \dqj \dqi
\\
\nonumber
& ~~~ \le \sum_i \int_{\qi \le \qri} (k-1)\bgxi (1 + \nu) \cdot \poti(\qi) \cdot \exi(\qi) \dqi
 +  (k-1)[(1 + \beta)^2 -1] \int_{\qi \le \qri} \poti(\qi)  \dqi \\
\nonumber
& \hspace{4.5in}\mbox{(using~\eqref{eqn:extreme-prob-bnd} and~\eqref{prob-large-win-bnd})}\\
\nonumber
&~~~~~~~~~  + \sum_i \sum_{j\ne i} \left[\int_{\qi \le \qri, \bqj \ge \bgxi} \poti(\qi) \cdot \xijs(\qi,\qj) \dqi\dqj \right. \\
\nonumber
& ~~~~~~~~~~~~~~~~~~~~~~~  - \left. \int_{\qj \le \qrj, \bqj \ge \bgxi} \potj(\qj) \cdot \txjis(\qj,\qi) \dqi\dqj \right] \\
\label{eqn:first-term-pot-diff}
& ~~~\le (k-1) \bgxi (1 + \nu) \sum_i \MRi + (k-1)[(1 + \beta)^2 -1] \sum_i \SRi \\
\label{eqn:scnd-term-pot-diff}
& ~~~~~~~~~ + \sum_i \sum_{j\ne i} \int_{\qi \le \qri, \qj \le \qrj, \bqj \ge \bgxi} [\poti(\qi) \cdot \xijs(\qi,\qj) - \potj(\qj) \cdot \txjis(\qj,\qi)] \dqi \dqj \\
\label{eqn:thrd-term-pot-diff}
& ~~~~~~~~~ + \sum_i\sum_{j\ne i} \int_{\qi \le \qri, \qj > \qrj, \bqj \ge \bgxi} \poti(\qi) \cdot \xijs(\qi,\qj) \dqi\dqj.
\end{align}

We bound \eqref{eqn:first-term-pot-diff}--\eqref{eqn:thrd-term-pot-diff} in turn.

\begin{align}
\nonumber
& (k-1) \bgxi (1 + \nu) \sum_i \MRi + (k-1)[(1 + \beta)^2 -1] \sum_i \SRi  \\
& ~~~~~~~~ \le  (k-1) \bgxi (1 + \nu) \MR + k(k-1)[(1 + \beta)^2 -1] \MR.
\end{align}

We can deduce from~\eqref{eq:i-sum-phi} that
\begin{equation}
\label{eq:i-sum-phi-ded}
\sum_{j \ne i} \int_{\bqj \ge \bgxi} \xijs(\qi,\qj) \dqj \le \exi(\qi) \le 1,
\end{equation}
and from~\eqref{eq:j-sum-phi}, for $\bqj \ge \bgxi$,
\begin{equation}
\label{eq:j-sum-phi-ded}
\sum_{i \ne j} \int_{\qi \ge 0} \xijs(\qi,\qj) \dqi = \sum_{i \ne j} \int_{\qi \ge 0} \txjis(\qj,\qi) \dqi \le \txj(\frac{\qj} {1 + \beta}) \le 1.
\end{equation}

For~\eqref{eqn:scnd-term-pot-diff},
recall that $\xijs(\qi,\qj) = \txjis(\qj,\qi) $.
Thus when $\bqj \ge \bgxij$,
if $\bqi \ge \bgxi (1 + \beta)$,
$ \poti(\qi) \cdot \xijs(\qi,\qj) - \potj(\qj) \cdot \txjis(\qj,\qi)
\le \xijs[2 \frac {\nu}{\Del} \frac {\SRj}{\bqj}
+ 2 \frac {\nu}{\Del} (1 + \beta) \cdot \frac {\SRi}{\bqi} ]$,
and if $\bqi < \bgxi (1 + \beta) $,
$\poti(\qi) \cdot \xijs(\qi,\qj) -  \potj(\qj) \cdot \txjis(\qj,\qi) \le
 \xijs[2 \frac {\nu}{\Del} \frac {\SRj}{\bqj} ]$.
Equivalently, if $\bqi \ge \bgxi (1 + \beta)$,
$ \poti(\qi) \cdot \xijs(\qi,\qj) -  \potj(\qj) \cdot \txjis(\qj,\qi)
\le \frac {\xijs(\qi,\qj)} {1 + \rho'} [2 \frac {\nu}{\Del} \frac {\SRj}{\bqj}
+ 2 \frac {\nu}{\Del} (1 + \beta) \cdot \frac {\SRi}{\bqi} ]$,
and if $\bqi < \bgxi (1 + \beta) $,
$\poti(\qi) \cdot \xijs(\qi,\qj) -  \potj(\qj) \cdot \txjis(\qj,\qi) \le
\frac {\xijs(\qi,\qj)} {1 + \rho'} [2 \frac {\nu}{\Del} \frac {\SRj}{\bqj} ]$.
Thus
\begin{align}
\nonumber
& \sum_i \sum_{j\ne i} \int_{\qi \le \qri, \qj \le \qrj, \bqj \ge \bgxi} [\poti(\qi) \cdot \xijs(\qi,\qj) - \potj(\qj) \cdot \txjis(\qj,\qi)] \dqi \dqj \\
\nonumber
& ~~~~ \le \sum_i \sum_{j\ne i} \int_{\qi \le \qri, \bqi \ge \bgxi (1 + \beta), \qj \le \qrj, \bqj \ge \bgxi}
 \frac {\xijs(\qi,\qj)} {1 + \rho'} \cdot \left[ 2 \frac {\nu}{\Del} \frac {\SRj}{\bqj}
 + 2  \frac {\nu}{\Del} (1 + \beta) \cdot \frac {\SRi}{\bqi} \right] \dqi \dqj \\
\nonumber
& ~~~~~~~~ + \sum_i \sum_{j\ne i} \int_{\qi \le \qri, \bqi < \bgxi (1 + \beta), \qj \le \qrj, \bqj \ge\bgxi}
 \frac {\xijs(\qi,\qj)} {1 + \rho'} \cdot  2  \frac {\nu}{\Del} \frac {\SRj}{\bqj} \dqi \dqj \\
\nonumber
& ~~~~ \le  \sum_i \sum_{j \ne i} \int_{\qi \le \qri, \qi \ge \frac{1 + \beta}{1 + \nu} \bgxi , \qj \le \qrj, \bqj \ge \bgxi}
2 \frac {\nu}{\Del} (1 + \beta)(1 + \nu) \cdot \frac {\SRi}{\qi} \cdot \xijs (\qi,\qj) \dqi \dqj\\
\label{eq:first-bound-on-small}
& ~~~~~~~~ +  \sum_j \int_{\qrj \ge \qj \ge \frac{\bgxi}{1 + \nu} }
2  \frac {\nu} {\Del} \frac{\SRj (1 + \nu)}{\qj} \dqj ~~~\mbox{(using~\eqref{eq:j-sum-phi-ded})}.
\end{align}

For~\eqref{eqn:thrd-term-pot-diff}, we note that as $\qj > \qrj$, $\potj(\qj) \le 0$, and then
when there is a small margin win by $i$, by definition,
if $\bgxi(1 + \beta) \le \bqi$,  $ \poti(\qi) \le 2 \frac {\nu} {\Del} \frac{\SRj} {\bqj}
+ 2 \frac {\nu} {\Del} (1 + \beta) \frac {\SRi} {\bqi}$,
and if $\bgxi(1 + \beta) > \bqi$, $ \poti(\qi) < 2 \frac {\nu} {\Del} \frac{\SRj} {\bqj}$.
Also, the constraint $\bqj \ge \bgxi$ implies $\qj \ge \frac{\bgxi}{(1 + \nu)}$.
Thus
\begin{align}
\nonumber
& \sum_i\sum_{j \ne i} \int_{\qi \le \qri, \qj > \qrj,  \bqj \ge \bgxi } \poti(\qi) \cdot \xijs(\qi,\qj) \dqi\dqj \\
\nonumber
& ~~~~ \le \sum_{j} \int_{\qj > \max\{\qrj,  \frac{\bgxi}{1 + \nu}\} } 2 \frac {\nu} {\Del}(1 + \nu) \frac{\SRj  } {\qj} \dqj ~~~\mbox{(using~\eqref{eq:j-sum-phi-ded}) } \\ 
 \label{eq:scnd-bound-on-small}
&~~~~~~~~ +  \sum_i \sum_{j \ne i} \int_{\qj \ge \qrj, \bqj \ge \bgxi, \bqi \ge \bgxi(1 + \beta)}  2 \frac {\nu} {\Del} (1 + \beta)(1 + \nu) \frac {\SRi} {\qi} \cdot \xijs (\qi,\qj) \dqi \dqj
\end{align}

Combining~\eqref{eq:first-bound-on-small} and~\eqref{eq:scnd-bound-on-small} yields
\begin{align}
\nonumber
& \sum_i \sum_{j\ne i} \int_{\qi \le \qri, \qj \le \qrj, \bqj \ge \bgxi} [\poti(\qi) \cdot \xijs(\qi,\qj) - \potj(\qj) \cdot \txjis(\qj,\qi)] \dqi \dqj \\
\nonumber
& ~~~~+\sum_i\sum_{j \ne i} \int_{\qi \le \qri, \qj > \qrj, \bqj \ge \bgxi} \poti(\qi) \xijs(\qi,\qj) \dqi\dqj \\
\nonumber
& ~~~~ \le  \sum_{j} \int_{\qj \ge \frac{\bgxi}{1 + \nu}} 2 \frac {\nu} {\Del} (1 + \nu)\frac{\SRj} {\qj} \dqj \\
\nonumber
&  ~~~~~~~~~~~~ +  \sum_i  \int_{\qi \ge \bgxi\frac{(1 + \beta)}{1 +\nu}}  2 \frac {\nu} {\Del} (1 + \beta)(1 + \nu) \frac {\SRi} {\qi} \dqi ~~~~~~ \mbox{(using~\eqref{eq:i-sum-phi-ded}) } \\
\nonumber
& ~~~~ \le \sum_i  (\rho + \rho') \MRi
+  \sum_j 2 \frac {\nu} {\Del} (1 + \beta) (1 + \nu) \SRi \ln \frac{(1 + \nu)}{\bgxi(1 + \beta)} \\
\nonumber
& ~~~~~~~~~~~~  + \sum_i 2 \frac {\nu} {\Del} \SRj (1 + \nu)\ln \frac{1+ \nu}{\bgxi} \\
& ~~~~ \le  4k \frac {\nu} {\Del} (1 + \beta) (1 + \nu) \ln \frac{1+ \nu}{\bgxi} \MR.
\end{align}

\end{proof}

 \hide{
We need one more bound.

\begin{lemma}
\label{lem:one-bidder-rev-loss}
\[
\int_{\qri < \qi \le \qbri} [-\poti(\qi)] \dqi \le [\nu(2 + \nu) + \bgxi] \cdot \SRi.
\]
\end{lemma}
\begin{proof}
As the integral is 0 when $\qri > \qbri$, we can safely assume $\qri \le \qbri$.
\begin{align*}
\int_{\qri < \qi \le \qbri} [-\poti(\qi)] \dqi
& = \ri \cdot \qri - \bri \cdot \qbri \\
&\le  \ri \cdot \qri - \frac{\bri \cdot \bqbri} {1 + \nu}.
\end{align*}
If $\qri \ge \xii$ then
\[
\bri \cdot {\bqbri} \ge \ri \cdot \bqri \ge \frac{\ri \cdot \qri} {1 + \nu}.
\]
Otherwise, using the concavity of the revenue function for the third inequality below,
\[
\bri \cdot \bqbri \ge \bvi(\bgxi) \cdot \bgxi \ge
\frac {\vi(\xii) \cdot \xii } {1 + \nu} \ge \frac{ \ri \cdot q(\ri)(1 - \xii)} {1 + \nu}
\ge \frac{ \ri \cdot q(\ri)(1 - \bgxi (1 + \nu))} {1 + \nu}.
\]
Thus
\begin{align*}
\int_{\qri < \qi \le \qbri} [-\poti(\qi)] \dqi
& \le \ri \cdot \qri - \min\left\{ \frac {\ri \cdot \qri} {(1 + \nu)^2}, \frac{ \ri \cdot \qri (1 - \bgxi(1 + \nu))} {(1 + \nu)^2} \right\}\\
& \le \ri \cdot \qri [\nu(2 + \nu) + \bgxi] =  [\nu(2 + \nu) + \bgxi]  \cdot \SRi.
\end{align*}
\end{proof}
}

We are now ready to bound $\Shtfl$.

\begin{lemma}
\label{lem:shtfl-bdd}
\begin{align*}
 \text{\rm{Shtf}}
&\le \MR \left[  k\del + k^2 \del + k\beta + k\bgxi(1 + \nu)  + 2k\nu
+ (k-1) \bgxi (1 + \nu)
+ k (k - 1) [(1 + \beta)^2 - 1 ]
 \right]
\\
 &~~~~+\MR\left[ 4 k(1 + \beta)(1 + \nu)  \frac {\nu}{\Del}\cdot \ln \frac{1 + \nu} {\bgxi} \right] .
\end{align*}
\end{lemma}
\begin{proof}
In the event that ${\cal E}_a$ does not hold for some $\Fi$,
which occurs with probability at most $k\del $,
the contribution to $\Shtfl$ is at most
\begin{align*}
& k\del  \left[ \sum_i \int_{\qi \le \qri} \poti(\qi) \exi(\qi) \dqi
- \int_{\qi \le \bqri} \poti(\qi) \txi(\qi) \dqi \right] \\
& \le k \del  \MR + k \del \sum _i \int_{\qri < \qi \le \bqri} [-\poti(\qi)] \dqi \\
& \le k \del \MR + k \del \sum _i\int_{\qi \le \qri} \poti(\qi) \dqi ~~~\mbox{(as $\int_{\qi} \poti(\qi) \dqi = 0$)}  \\
& \le  k \del \MR + k \del \sum_i \SRi \le  k \del \MR + k^2 \del  \MR.
\end{align*}
Otherwise, the contribution is given by summing the bounds from ~\eqref{eqn:shtfl-frst-bnd},
Lemmas~\ref{lem:frth-term-i}--\ref{lem:fifth-term}, and~\ref{lem:scnd-term}.
\end{proof}

\hide{
\begin{theorem}
\label{thm:emp-myer-perf}
In the empirical Myerson auction with $k$ bidders each having a regular distribution,
using $m$ independent samples from its distribution for each bidder,
the resulting expected revenue satisfies $\EMR \ge (1 - \eps) \MR$ if
\[
m = \frac {k^{11}}{\eps^7} (\ln^3 k + \ln^3 \frac {1} {\eps}) +
\frac {k^8}{\eps^5} (\ln^2 k + \ln^2 \frac {1} {\eps})\ln \frac{1} {\del}.
\]
\end{theorem}
}

\begin{pfof}{Theorem~\ref{thm:emp-myer-perf}}
We first choose $\Del, \nu, \bgxi \le\frac{1}{12}$.
It is easy to check that then $(1 + \beta)^2 -1 = (1 + \Del)^2(1 + \nu)^6 -1
\le 2\Del(1 +\Del)(1 + \nu)^6 + 6\nu(1+\Del)^2(1+\nu)^5
\le (2 \Del + 6\nu)
\left( \frac {13} {12} \right)^7
\le 4 \Del + 11\nu$.
Similarly, $\beta \le 2 \Del + 4\nu$, and
$4 (1 + \beta)(1 + \nu) \le 4 \left( \frac {13}{12} \right)^5 \le 4 \cdot \tfrac 32 = 6$.
Consequently,
\begin{equation*}
\text{\rm{Shtf}}
\le \MR \left[  k\del +k^2\del + 2k\Del + 4k\nu + 3 k \bgxi + 2k \nu + 4 k (k - 1) \Del
+11k (k - 1) \nu  + 6k \frac {\nu}{\Del} \ln \frac {(1+ \nu)^2}{\bgxi} \right].
\end{equation*}
It suffices that $\Shtfl \le \eps \MR$.
To this end, we bound the right hand side of the above expression
by $\eps$.
To achieve this it suffices to choose $\nu$, $\bgxi$, $\del$, and $\Del$ as follows:
\begin{eqnarray*}
 k (k + 1) \del & =  & \frac 14 \eps \\
 (4 k (k - 1) +2k)  \Del & =  & \frac 14 \eps \\
3k\bgxi & = & \frac 14 \eps \\
4k\nu + 2k \nu +  11 k (k - 1) \nu + 6k  \frac {\nu} {\Del} \ln \frac {(1+ \nu)^2}{\bgxi} & =  &  \frac 14 \eps.
\end{eqnarray*}

It suffices that
\begin{eqnarray*}
 \del & =  & \Theta \left(\frac{\eps} {k^2} \right) \\
 \Del & =  & \Theta \left(\frac{\eps} {k^2} \right) \\
\bgxi & =  & \Theta \left(\frac{\eps} {k}\right) \\
\nu & =  & \Theta \left(\frac{\eps^2} {k^3} \frac{1} {\ln \frac {k} {\eps} } \right).
\end{eqnarray*}

One final detail is that we need to set $\hgxi$ also, but it suffices to note
that $\hgxi = \bgxi + \frac {1}{2m}$ and so
$\hgxi = \Theta(\frac{\eps} {k}) + \frac {1}{2m}$.

By Lemma~\ref{lem:samp-acc},
$m = \Omega(\frac{1}{\gam^3 \hgxi} + \ln \frac{1} {\del} \cdot \frac {1}{\gam^2 \hgxi})$ suffices.
Recalling that $1 + \nu = (1 + \gam)^2$, so $\gam = \Theta(\nu)$,
we obtain that
$m = \Omega(\frac {k^{10}}{\eps^7} \ln^3 \frac {k} {\eps})$
suffices.
\end{pfof}

\section{Conclusions}

This paper proposes a general model for learning a near-optimal
auction from data, in the form of i.i.d.\ samples from unknown
distributions.  It provides upper and lower bounds on the sample
complexity for the case of single-item auctions, and
shows that the number of samples required to obtain a
$(1-\eps)$-approximation of the optimal expected revenue scales
polynomially with both the number~$k$ of bidders and with
$\tfrac{1}{\eps}$.  We conclude by listing some of the many directions
in which this work could be extended.
\begin{enumerate}

\item Prove tight bounds (in terms of~$k$ and~$\tfrac{1}{\eps}$) on
  how many  samples are necessary and   sufficient to achieve expected
  revenue at least $1-\eps$ times the maximum possible.
  (See~\cite{HMR15,MR15,DHP15} for recent progress.)

\item Prove good sample complexity upper bounds for settings other
  than single-item   auctions.  (See~\cite{MR15} for recent progress.)

\item Are there natural settings where the learning problem is
  information-theoretically easy (meaning polynomial sample
  complexity) yet computationally hard (under complexity assumptions)?

\item Identify multi-parameter problems, less general than those
  in~\cite{dughmi2014sampling}, where a near-optimal mechanism can be
  learned from a polynomial number of samples.

\item For problems where a $(1-\eps)$-approximate mechanism
cannot be learned from a polynomial number of samples, identify the
  best-possible approximation factor for which an approximately
  optimal mechanism can be efficiently learned.

\item If some of the bidders that contribute the samples are the same
  as the bidders that participate in the final auction, then these
  bidders need not bid truthfully.  (Underbidding could result in
  lower payments in the future.)  Is it still possible to learn a
  near-optimal auction in this setting?

\end{enumerate}

\bibliographystyle{plain}
\bibliography{auction-refs}

\end{document}